\documentclass[a4paper,UKenglish]{lipics-v2019}
\nolinenumbers
\usepackage{microtype}

\usepackage{amsmath,amsthm,amssymb}
\makeatother \makeatletter
\newif\if@restonecol
\makeatother

\usepackage[linesnumbered,ruled,vlined]{algorithm2e}

\usepackage{graphicx}
\usepackage{color}
\usepackage{colortbl}
\usepackage{url}
\usepackage{verbatim}
\usepackage[shortlabels]{enumitem}
\usepackage{amssymb}
\usepackage{multirow}
\usepackage{balance}
\usepackage{lscape}
\usepackage{float}
\usepackage[T1]{fontenc}
\usepackage{tgtermes}

\usepackage{url}
\usepackage{enumitem}

\setlist[itemize]{leftmargin=*}
\setlist[enumerate]{leftmargin=*}

\definecolor{mygray}{gray}{.9}
\definecolor{mygreen}{rgb}{0.0, 0.5, 0.0}
\definecolor{myred}{rgb}{0.8, 0.0, 0.0}
\definecolor{mygreen1}{rgb}{0.03, 0.91, 0.87}
\definecolor{mypink}{rgb}{1.0, 0.0, 0.5}
\definecolor{mycorn}{rgb}{0.98, 0.93, 0.36}

\newtheorem{assumption}{Assumption}
\newtheorem{problem}{Problem}

\newcommand{\yuliang}[1]{{\it\small\textcolor{blue}{[[[ {#1}\ --yuliang ]]]}}}

\newcommand{\yannis}[1]{{\it\small\textcolor{myred}{[[[ {#1}\ --Yannis ]]]}}}

\newcommand{\victor}[1]{{\it\small\textcolor{red}{[[[ {#1}\ --victor ]]]}}}

\renewcommand{\yuliang}[1]{}
\renewcommand{\yannis}[1]{}
\renewcommand{\victor}[1]{}

\DeclareMathOperator*{\argmin}{\arg\!\min}
\DeclareMathOperator*{\argmax}{\arg\!\max}

\usepackage{dsfont}

\newcommand{\ma}{\mathbf{a}}
\newcommand{\mb}{\mathbf{b}}

\newcommand{\ms}{\mathbf{s}}
\newcommand{\me}{\mathds{1}}
\newcommand{\mq}{\mathbf{q}}
\newcommand{\mC}{\mathbf{C}}

\newcommand{\db}{\mathcal{D}}

\newcommand{\csim}{\mathtt{cos}}

\newcommand{\calL}{\mathcal{L}}
\newcommand{\calO}{\mathcal{O}}
\newcommand{\calT}{\mathcal{T}}

\newcommand{\calB}{\mathcal{B}}
\newcommand{\opt}{\mathsf{OPT}}
\newcommand{\MR}{\mathsf{MR}}

\newcommand{\HL}{\mathsf{HL}}
\newcommand{\LQ}{\mathsf{LQ}}
\newcommand{\LL}{\mathsf{L2}}
\newcommand{\QQ}{\mathsf{Q2}}

\newcommand{\TT}{\mathsf{T}}
\newcommand{\hull}{\mathsf{H}}

\newcommand{\TA}{\ensuremath{\mathsf{TA}}}
\newcommand{\BL}{\ensuremath{\mathsf{BL}}}
\newcommand{\TC}{\ensuremath{\mathsf{TC}}}
\newcommand{\score}{\ensuremath{\mathsf{MS}}}

\newcommand{\cost}{\ensuremath{\mathsf{cost}}}
\newcommand{\stp}{\ensuremath{\mathsf{stop}}}

\title{Index-based, High-dimensional, Cosine Threshold Querying with Optimality Guarantees}

\author{Yuliang Li}{Megagon Labs \& UC San Diego, San Diego, California, USA}{}{}{}
\author{Jianguo Wang}{UC San Diego, San Diego, California, USA}{}{}{}
\author{Benjamin Pullman}{UC San Diego, San Diego, California, USA}{}{}{}
\author{Nuno Bandeira}{UC San Diego, San Diego, California, USA}{}{}{}
\author{Yannis Papakonstantinou}{UC San Diego, San Diego, California, USA}{}{}{}

\authorrunning{Y. Li, J. Wang, B. Pullman, N. Bandeira, and Y. Papakonstantinou}
\Copyright{Yuliang Li, Jianguo Wang, Benjamin Pullman, Nuno Bandeira, and Yannis Papakonstantinou}



\EventEditors{Pablo Barcelo and Marco Calautti}
\EventNoEds{2}
\EventLongTitle{22nd International Conference on Database Theory (ICDT 2019)}
\EventShortTitle{ICDT 2019}
\EventAcronym{ICDT}
\EventYear{2019}
\EventDate{March 26--28, 2019}
\EventLocation{Lisbon, Portugal}
\EventLogo{}
\SeriesVolume{127}
\ArticleNo{8} 

\begin{document}

\maketitle
\begin{abstract}
Given a database of vectors, a cosine threshold query returns all vectors in the database
having cosine similarity to a query vector above a given threshold $\theta$.
These queries arise naturally in many applications, 
such as document retrieval, image search, and mass spectrometry. 
The present paper considers the efficient evaluation of such queries, providing novel optimality
guarantees and exhibiting good performance on real datasets.
We take as a starting point Fagin's well-known Threshold Algorithm ($\TA$), 
which can be used to answer cosine threshold queries as follows: an inverted index is first built from the database vectors during pre-processing;
at query time, the algorithm traverses the index partially to gather a set of candidate vectors 
to be later verified for $\theta$-similarity. However, directly applying $\TA$ in its raw form 
misses significant optimization opportunities. 
Indeed, we first show that one can take advantage of the fact that the vectors can be assumed to be normalized, 
to obtain an improved, tight stopping condition for index traversal and to efficiently compute it incrementally. 
Then we show that one can take advantage of data skewness to obtain better traversal strategies.
In particular, we show a novel traversal strategy that exploits a common data skewness condition which holds in multiple 
domains including mass spectrometry, documents, and image databases.
We show that under the skewness assumption, the new traversal strategy has a strong, near-optimal performance guarantee. 
The techniques developed in the paper are quite general since they can be applied to a large class of similarity functions beyond cosine. 
\end{abstract}

\vspace{-2mm}
\section{Introduction}
\label{sec:intro}

Given a database of vectors, a cosine threshold query asks for all database vectors with
cosine similarity to a query vector above a given threshold.

This problem arises in many applications including
document retrieval~\cite{Broder2003EQE}, image search~\cite{Kulis09}, recommender systems~\cite{Li2017FFE} and mass spectrometry.
For example, in mass spectrometry, billions of spectra are generated for the purpose of protein analysis~\cite{Aebersold16,PMIC200600625,pr400230p}.
Each spectrum is a collection of key-value pairs where the key is the mass-to-charge ratio
of an ion contained in the protein and the value is the intensity of the ion.
Essentially, each spectrum is a high-dimensional, non-negative and sparse vector with $\sim$2000 dimensions where $\sim$100 coordinates are non-zero.

Cosine threshold queries play an important role in analyzing such spectra repositories.
Example questions include ``is the given spectrum similar to any spectrum in the database?'',
spectrum identification (matching query spectra against reference spectra), or
clustering (matching pairs of unidentified spectra) or
metadata queries (searching for public datasets containing matching spectra, even if obtained from different types of samples).
For such applications with a large vector database,
it is critically important to process cosine threshold queries efficiently --
this is the fundamental topic addressed in this paper.

\vspace{-1.5mm}
\begin{definition}[Cosine Threshold Query]
Let $\db$ be a collection of high-dimensional, non-negative vectors; $\mq$ be a query vector; $\theta$ be a threshold $0 < \theta \leq 1$. Then the cosine threshold query  returns the vector set
$\mathcal{R} = \{\ms | \ms \in \db, \csim(\mq, \ms) \ge \theta \}$. %
A vector $\ms$ is called \emph{$\theta$-similar} to the query $\mq$ if $\csim(\mq, \ms) \geq \theta$ and the \emph{score} of $\ms$ is the value $\csim(\mq, \ms)$ when
$\mq$ is understood from the context.
\vspace{-1.5mm}
\end{definition}

Observe that cosine similarity is insensitive to vector normalization.
We will therefore assume without loss of generality that the database as well as query consist of unit vectors (otherwise, all vectors can be normalized in a pre-processing step).

In the literature, cosine threshold querying is a special case of Cosine Similarity Search (CSS)~\cite{Teflioudi2016,AnastasiuK14,Li2017FFE}, where
other aspects like approximate answers, top-k queries and similarity join are considered.
Our work considers specifically CSS with exact, threshold and single-vector queries, which is the case of interest to many applications.

Because of the unit-vector assumption, the scoring function $\csim$ computes 
the dot product $\mq \cdot \ms$.
Without the unit-vector assumption, the problem is equivalent to \emph{inner product threshold querying}, which is of interest in its own right. 
Related work on cosine and inner product similarity search is summarized in Section \ref{sec:related}.


In this paper we develop novel techniques for the efficient evaluation of cosine threshold queries. 
We take as a starting point the well-known Threshold Algorithm (\textsf{TA}), 
by Fagin et al.~\cite{Fagin2001OAA}, because of its simplicity, wide applicability, and optimality guarantees. 
We review the classic $\TA$ in Appendix \ref{app:ta}.

\smallskip
\noindent \textbf{A \textsf{TA}-like baseline index and algorithm and its shortcomings.}
The $\TA$ algorithm can be easily adapted to our setting, yielding a first-cut approach to processing cosine threshold queries. We describe how this is done
and refer to the resulting index and algorithm as the {\em \textsf{TA}-like baseline}.
Note first that cosine threshold queries use $\csim(\mq,\ms)$, which can be viewed as a particular family of functions $F(\ms) = \ms \cdot \mq$ parameterized by $\mq$, that are monotonic in $\ms$ for unit vectors. However, \textsf{TA} produces the vectors with the top-k scores according to $F(\ms)$, whereas
cosine threshold queries return all $\ms$ whose score exceeds the threshold $\theta$. We will show how this difference can be overcome straightforwardly.

A \emph{baseline} index and algorithm inspired by $\TA$ can answer cosine threshold queries exactly without a full scan of the vector database for each query.
In addition, the baseline algorithm enjoys the same instance optimality guarantee as the original $\TA$. 
This baseline is created as follows.
First, identically to the \textsf{TA}, the baseline index consists of one sorted list for each of the $d$ dimensions. In particular, the $i$-th sorted list has pairs $(\mathsf{ref}(\ms), \ms[i])$, where $\mathsf{ref}(\ms)$ is a reference to the vector $\ms$ and $\ms[i]$ is its value on the $i$-th dimension. The list is sorted in descending order of $\ms[i]$.%
\footnote{There is no need to include pairs with zero values in the list.}

Next, the baseline, like the \textsf{TA}, proceeds into a \textit{gathering phase} during which it collects a complete set of references to candidate result vectors. The \textsf{TA} shows that gathering can be achieved by reading the $d$ sorted lists from top to bottom and terminating early when a \textit{stopping condition} is finally satisfied. The condition guarantees that any vector that has not been seen yet has no chance of being in the query result. 
The baseline makes a straightforward change to the \textsf{TA}'s stopping condition to adjust for the difference between the \textsf{TA}'s top-k requirement and the threshold requirement of the cosine threshold queries. In particular, in each round the baseline algorithm has read the first $b$ entries of each index. (Initially it is $b=0$.) If it is the case that $\csim(\mq, [L_1[b], \ldots, $ $L_d[b]]) <\theta$ then it is guaranteed that the algorithm has already read (the references to) all the possible candidates and thus it is safe to terminate the gathering phase, see Figure~\ref{fig:specExample} for an example. \victor{very nice example!} Every vector $\ms$ that appears in the $j$-th entry of a list for $j<b$ is a candidate.

In the next phase, called the \textit{verification} phase, the baseline algorithm (again like \textsf{TA}) retrieves the candidate vectors from the database and checks which ones actually score above the threshold. 



\begin{figure}[tbp]
\centering
\renewcommand{\tabcolsep}{0.1mm}
\includegraphics[width=0.9\textwidth]{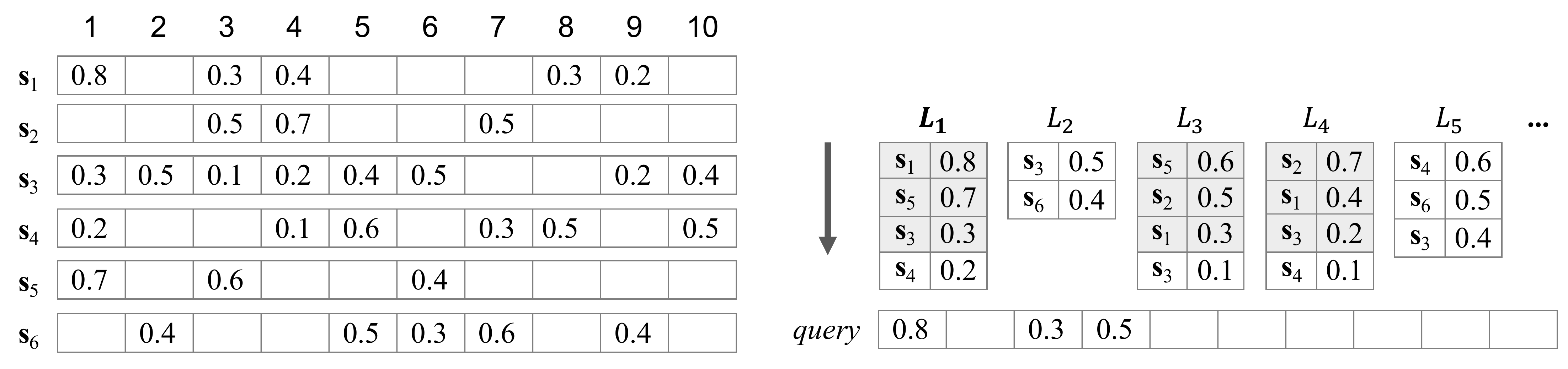}
\caption{{An example of cosine threshold query with six 10-dimensional vectors. The missing values are 0's. We only need to scan the lists $L_1$, $L_3$, and $L_4$ since the query vector has non-zero values in dimension 1, 3 and 4. For $\theta = 0.6$, the gathering phase terminates after each list has examined three entries (highlighted) because the score for any unseen vector is at most $0.8\times 0.3 + 0.3\times 0.3 + 0.5\times 0.2 = 0.43 < 0.6$. The verification phase only needs to retrieve from the database those vectors obtained during the gathering phase, i.e., $\ms_1$, $\ms_2$, $\ms_3$ and $\ms_5$, compute the cosines and produce the final result.}}\label{fig:specExample}
\end{figure}

For inner product queries, 
the baseline algorithm's gathering phase benefits from the same $d \cdot \opt$ instance optimality guarantee as the \textsf{TA}. Namely, the gathering phase will access at most $d \cdot \opt$ entries, where $\opt$ is the optimal index access cost. 
More specifically, the notion of $\opt$ is the \emph{minimal number of sequential accesses} 
of the sorted inverted index during the gathering phase for any $\TA$-like algorithm applied to the specific query and index instance.

There is an obvious optimization: Only the $k$ dimensions that have non-zero values in the query vector $\mq$ should participate in query processing -- this leads to a $k \cdot \opt$ guarantee for inner product queries.\footnote{This optimization is equally applicable to the \textsf{TA}'s problem: Scan only the lists that correspond to dimensions that actually affect the function $F$.}
But even this guarantee loses its practical value when $k$ is a large number. In the mass spectrometry scenario $k$ is $\sim$100. In document similarity and image similarity cases it is even higher.

For cosine threshold queries, the $k \cdot \opt$ guarantee no longer holds.
The baseline fails to utilize the unit vector constraint to reach the stopping condition faster,
resulting in an unbounded gap from $\opt$ because of the unnecessary accesses (see Appendix \ref{sec:TANotTight}).\footnote{Notice, the unit vector constraint enables inference about the collective weight of the unseen coordinates of a vector.}
Furthermore, the baseline fails to utilize the skewing of the values in the vector's coordinates
(both of the database's vectors and of the query vector) and the linearity of the similarity function. Intuitively, if the query's weight is concentrated on a few coordinates, the query processing should overweight the respective lists and may, thus, reach the stopping condition much faster than reading all relevant lists in tandem.

We retain the baseline's index and the gathering-verification structure which characterizes the family of $\TA$-like algorithms.
The decision to keep the gathering and verification stages separate is discussed in Section \ref{sec:framework}. We argue that this algorithmic structure is appropriate for cosine threshold queries, because further optimizations that would require merging the two phases are only likely to yield marginal benefits. 
Within this framework, we reconsider 
\vspace{-2mm}
\begin{enumerate}\parskip=0pt
\item \textit{Traversal strategy optimization}: A \textit{traversal strategy} determines the order in which the gathering phase proceeds in the lists. In particular, we allow the gathering phase to move deeper in some lists and less deep in others. For example, the gathering phase may have read at some point $b_1=106$ entries from the first list, $b_2=523$ entries from the second list, etc. Multiple traversal strategies are possible and, generally, each traversal strategy will reach the stopping condition with a different configuration of $[b_1, b_2, \ldots, b_n]$. The traversal strategy optimization problem asks that we efficiently identify a traversal path that minimizes the access cost $\sum_{i=1}^{d}b_i$.
To enable such optimization, we will allow lightweight additions to the baseline index.
\item \textit{Stopping condition optimization}: We reconsider the stopping condition so that it takes into account (a) the specifics of the $\csim$ function and (b) the unit vector constraint.
Moreover, since the stopping condition is tested frequently during the gathering phase, it has to to be evaluated very efficiently.
Notice that optimizing the stopping condition is independent of the traversal strategy or skewness assumptions about the data.
\end{enumerate}
\vspace{-2mm}

\setlength{\tabcolsep}{3.5pt}
\begin{small}
\begin{table}[!t]
\small
\caption{Summary of theoretical results for the near-convex case.}
\centering
  \begin{tabular}{l||c|c||c|c}\hline\hline
    \multirow{2}{*}{ } & \multicolumn{2}{c||}{\textbf{Stopping Condition}} & \multicolumn{2}{c}{\textbf{Traversal Strategy}}\\\cline{2-5}
      & \emph{Baseline} & \emph{This work} & \emph{Baseline} & \emph{This work} \\\hline
Inner Product & \multicolumn{2}{c||}{Tight} & $m \cdot \opt$ & \cellcolor[gray]{0.9} $\opt + c$ \\\hline
Cosine        & Not tight & \cellcolor[gray]{0.9} Tight & $\mathsf{NA}$ & \cellcolor[gray]{0.9} $\opt(\theta - \epsilon) + c$ \\\hline\hline
  \end{tabular}
\label{tbl:result}
\vspace{-6mm}
\end{table}
\end{small}

\smallskip 
\noindent 
\textbf{Contributions and summary of results. }
\vspace{-2mm}
\begin{itemize}\parskip=0pt
\item We present a stopping condition for early termination of the index traversal (Section \ref{sec:stop}). We show that the stopping condition is \emph{complete} and \emph{tight}, informally meaning that (1) for any traversal strategy, the gathering phase will produce a candidate set containing all the vectors $\theta$-similar to the query,
and (2) the gathering terminates as soon as no more $\theta$-similar vectors can be found (Theorem \ref{thm:qiLibi}).
In contrast, the stopping condition of the (\textsf{TA}-inspired) baseline is complete but not tight (Theorem \ref{thm:TAnotTight}, Appendix \ref{sec:TANotTight}).
The proposed stopping condition takes into account that all database vectors are normalized and reduces the problem to solving a
special quadratic program (Equation \ref{equ:quadratic}) that guarantees both completeness and tightness.
While the new stopping condition prunes further the set of candidates, it can also be efficiently computed in $\calO(\log d)$ time using incremental maintenance techniques.
%
\item We introduce the \emph{hull-based} traversal strategies that exploit the skewness of the data (Section \ref{sec:hull-based}).
In particular, skewness implies that each sorted list $L_i$ is ``mostly convex'', meaning that
the shape of $L_i$ is approximately the \textit{lower convex hull} constructed from the set of points of $L_i$.
This technique is quite general, as it can be extended to the class of \emph{decomposable functions} which have the form
$ F(\ms) = f_1(\ms[1]) + \ldots + f_d(\ms[d]) $
where each $f_i$ is non-decreasing.\footnote{The inner product threshold problem is the special case where $f_i(\ms[i])=q_i\cdot \ms[i]$.}
Consequently, we provide the following optimality guarantee for inner product threshold queries:
The number of accesses executed by the gathering phase (i.e., $\sum_{i=1}^{d}b_i$) is at most $\opt+c$
(Theorem \ref{thm:near-optimal} and Corollary \ref{cor:near-optimal}), where $\opt$ is the number of accesses by the optimal strategy and
$c$ is the max distance between two vertices in the lower convex hull.
Experiments show that in multiple real-world cases, $c$ is a very small fraction of $\opt$. 
\item Despite the fact that cosine and its tight stopping condition are not decomposable,
we show that the hull-based strategy can be adapted to cosine threshold queries
by approximating the tight stopping condition with a carefully chosen decomposable function. 
We show that when the approximation is at most $\epsilon$-away from the actual value,
the access cost is at most $\opt(\theta - \epsilon) + c$ (Theorem \ref{thm:near-optimal-approx})
where $\opt(\theta - \epsilon)$ is the optimal access cost
on the same query $\mq$ with the threshold lowered by $\epsilon$ and $c$ is a constant similar to the above decomposable cases.
Experiments show that the adjustment $\epsilon$ is very small in practice, e.g., 0.1.
We summarize these new results in Table \ref{tbl:result}.
\end{itemize}
\vspace{-2mm}

The paper is organized as follows. We introduce the algorithmic framework and basic definitions in Section \ref{sec:framework}.
Section \ref{sec:stop} and \ref{sec:hull-based} discuss the technical developments as we mentioned above.
Finally, we discuss related work in Section~\ref{sec:related} and conclude in Section~\ref{sec:conclusion}.

\vspace{-2mm}
\section{Algorithmic Framework}\label{sec:framework}



In this section, we present a Gathering-Verification algorithmic framework to facilitate optimizations in different
components of an algorithm with a $\TA$-like structure. We start with notations summarized in Table~\ref{tbl:symbol}.


To support fast query processing, we build an index for the database vectors similar to the original \TA.
The basic index structure consists of a set of 1-dimensional sorted lists (a.k.a inverted lists in web search \cite{Broder2003EQE})
where each list corresponds to a vector dimension and contains vectors having non-zero values on that dimension,
as mentioned earlier in Section \ref{sec:intro}. Formally, for each dimension $i$, $L_i$ is a list of pairs
$\{(\mathsf{ref}(\ms),\ms[i]) \ | \ \ms \in \db \land \ms[i] > 0\}$ sorted in descending order of $\ms[i]$ where $\mathsf{ref}(\ms)$ is a reference to the vector $\ms$ and $\ms[i]$ is its value on the $i$-th dimension. 
In the interest of brevity, we will often write $(\ms, \ms[i])$ instead of $(\mathsf{ref}(\ms),\ms[i])$.
As an example in Figure~\ref{fig:specExample}, the list $L_1$ is built for the first dimension and
it includes 4 entries: $(\ms_1,0.8)$, $(\ms_5,0.7)$, $(\ms_3,0.3)$, $(\ms_4,0.2)$ because $\ms_1$, $\ms_5$, $\ms_3$ and $\ms_4$
have non-zero values on the first dimension.

Next, we show the Gathering-Verification framework (Algorithm~\ref{alg:traverse-verify}) that operates on the index structure.
The framework includes two phases: the gathering phase and the verification phase. 

\begin{table}
\caption{Notation}
\vspace{-2mm}
\label{tbl:symbol}
\centering
\makebox[0pt][c]{\parbox{1.1\textwidth}{%
    \begin{minipage}[b]{0.5\hsize}\centering
    \begin{small}
  \begin{tabular}{l|lll}
    \hline\hline
    $\db$ & the vector database \\\hline
    $d$ & the number of dimensions \\\hline
    $\ms$ (bold font) & a data vector\\\hline
    $\mq$ (bold font) & a query vector\\\hline
    $\ms[i]$ or $s_i$ & the $i$-th dimensional value of $\ms$\\\hline
    $|\ms|$ & the L1 norm of $\ms$ \\\hline
    $\| \ms \|$ & the L2 norm of $\ms$ \\\hline
    $\theta$ & the similarity threshold \\\hline
    $\csim(\mathbf{p},\mq)$ & the cosine of vectors $\mathbf{p}$ and $\mq$\\\hline
    $L_i$ & \begin{tabular}{@{}l@{}}the inverted list of the $i$-th \\  dimension \end{tabular}  \\\hline
    $\mb = (b_1, \dots, b_d)$ & a position vector \\\hline
    $L_i[b_i]$ & the $b_i$-th value of $L_i$\\\hline
    $L[\mb]$ & the vector $(L_1[b_1], \dots, L_d[b_d])$\\\hline
    \hline
  \end{tabular}\end{small}
  \end{minipage}
    \hfill
    \begin{minipage}[b]{0.5\hsize}\centering
\begin{small}
\begin{algorithm}[H]\small
\SetKwInOut{Input}{input}
\SetKwInOut{Output}{output}
\SetKwInOut{Variables}{variables}
\Input{ ($\db$, $\{L_i\}_{1 \leq i \leq d}$, $\mq$, $\theta$)}
\Output{ $\mathcal{R}$ the set of $\theta$-similar vectors} 
\tcc{Gathering phase}
Initialize $\mb = (b_1, \dots, b_d) = (0, \dots, 0)$\;
\tcp{$\varphi(\cdot)$ is the stopping condition}
\While{$\varphi(\mb) = \mathtt{false}$}
{
\tcp{$\calT(\cdot)$ is the traversal strategy to determine which list to access next}
    $i \leftarrow \calT(\mb)$\;
    $b_i \leftarrow b_i + 1$\;
    Put the vector $\ms$ in $L_i[b_i]$ to the candidate pool $\mathcal{C}$\;
}
\tcc{Verification phase}
$\mathcal{R} \leftarrow \{ \ms | \ms \in \mathcal{C} \land \csim(\mq, \ms) \geq \theta \}$\;
\Return $\mathcal{R}$\;
 \caption{\small{\textsf{Gathering-Verification Framework}}}\label{alg:traverse-verify}
\end{algorithm}\end{small}
    \end{minipage}%
}}
\vspace{-1mm}
\end{table}

\smallskip
\noindent \textbf{Gathering phase} (line 1 to line 5). The goal of the gathering phase is to collect a complete set of candidate vectors while minimizing the number of accesses to the sorted lists. The algorithm maintains a \emph{position vector} $\mb = (b_1,\dots, b_d)$ where each $b_i$ indicates the current position in the inverted list $L_i$.
Initially, the position vector $\mb$ is $(0,\dots,0)$. Then it traverses the lists according to a \emph{traversal strategy}
that determines the list (say $L_i$) to be accessed next (line 3).  Then it advances the pointer $b_i$ by 1 (line 4) and adds the vector $\ms$ 
referenced in the entry $L_i[b_i]$ to a candidate pool $\mathcal{C}$ (line 5).
The traversal strategy is usually stateful, which means that its decision is made based on
information that has been observed up to position $\mb$ and its past decisions.
For example, a strategy may decide that it will make the next 20 moves along dimension 6 and thus it needs state in order to remember that it has already committed to 20 moves on dimension 6.

The gathering phase terminates once a \emph{stopping condition} is met.
Intuitively, based on the information that has been observed in the index,
the stopping condition checks if a complete set of candidates has already been found.

Next, we formally define stopping conditions and traversal strategies.
As mentioned above, the input of the stopping condition and the traversal strategy is the information that has been observed up to position $\mb$,
which is formally defined as follows.
\vspace*{-2mm}
\begin{definition}\label{def:partial}
Let $\mb$ be a position vector on the inverted index $\{L_i\}_{1\leq i \leq d}$ of a database $\db$.
The partial observation at $\mb$, denoted as $\calL(\mb)$, is a collection of lists
$\{\hat{L}_i\}_{1 \leq i \leq d}$ where for every $1 \leq i \leq d$, $\hat{L_i} = [L_i[1], \dots, L_i[b_i]]$.
\end{definition}
\vspace*{-2mm}
\begin{definition}\label{def:stopping-condition}
Let $\calL(\mb)$ be a partial observation and $\mathbf{q}$ be a query with similarity threshold $\theta$.
A \textbf{stopping condition} is a boolean function $\varphi(\calL(\mb), \mq, \theta)$
and a \textbf{traversal strategy} is a function $\calT(\calL(\mb), \mq, \theta)$ whose domain is $[d]$\footnote{$[d]$ is the set $\{1, \dots, d\}$}.
When clear from the context, we denote them simply by $\varphi(\mb)$ and $\calT(\mb)$ respectively.
\end{definition}
\vspace*{-2mm}

\smallskip
\noindent \textbf{Verification phase} (line 6).
The verification phase examines each candidate vector $\ms$ seen in the gathering phase to verify whether $\csim(\mq,\ms)\ge\theta$ by accessing the database.
Various techniques \cite{Teflioudi2016,Anastasiu2015PFP,Li2017FFE} have been proposed to speed up this process.
Essentially, instead of accessing all the $d$ dimensions of each $\ms$ and $\mq$ to compute exactly the cosine similarity,
these techniques decide $\theta$-similarity by performing a partial scan of each candidate vector.
We review these techniques, which we refer to as \emph{partial verification}, in Appendix \ref{sec:verification}.
Additionally, as a novel contribution, we show that in the presence of data skewness,
partial verification can have a near-constant performance guarantee (Theorem \ref{thm:verification}) for each candidate.
\vspace{-2mm}
\begin{theorem}\label{thm:verification-informal}
(Informal) For most skewed vectors, $\theta$-similarity can be computed at constant time.
\vspace{-2mm}
\end{theorem}

\smallskip
\noindent
\textbf{Remark on optimizing the gathering phase. }
Due to these optimization techniques, 
the \emph{number of sequential accesses} performed during the gathering phase becomes
the dominating factor of the overall running time.
This reason behind is that the number of sequential accesses is strictly greater than 
the number of candidates that need to be verified so
reducing the sequential access cost also results in better performance of the verification phase. 
In practice, we observed that the sequential cost is indeed dominating: 
for 1,000 queries on 1.2 billion vectors with similarity threshold 0.6, 
the sequential gathering time is 16 seconds and the verification time is only 4.6 seconds.
Such observation justifies 
our goal of designing a traversal strategy with near-optimal sequential access cost, 
as the dominant cost concerns the gathering stage.

\smallskip
\noindent \textbf{Remark on the suitability of $\TA$-like algorithms.} 
One may wonder whether algorithms that start the gathering phase NOT from the top of the inverted lists may outperform the best $\TA$-like algorithm.
In particular, it appears tempting to start the gathering phase from the point closest to $q_i$ in each inverted list
and traverse towards the two ends of each list. 
Appendix~\ref{app:middle} proves why this idea can lead to poor performance. In particular, Appendix~\ref{app:middle} proves that in a general setting, 
the computation of a tight and complete stopping condition (formally defined in Definition \ref{def:complete} and \ref{def:tight}) becomes {\sc np-hard} since it needs to take into account constraints from two pointers (forward and backward) for each inverted list.
Furthermore, in many applications, the data skewing leads to small savings from pruning the top area of each list, since the top area is sparsely populated - unlike the densely populated bottom area of each list. Thus it is not justified to use an expensive gathering phase algorithm for small savings. 
\yuliang{in applications like mass spec => in many applications}

Section \ref{sec:relatedtechniques} reviews additional prior work ideas \cite{Teflioudi2016,Teflioudi2015LFR} that 
avoid  traversing some top/bottom regions of the inverted index.
Such ideas may provide additional optimizations to $\TA$-like algorithms in variations and/or restrictions of the problem (e.g., a restriction that the threshold is very high)
and thus they present future work opportunities in closely related problems.



\vspace{-2mm}
\section{Stopping condition}\label{sec:stop}

%
%
In this section, we introduce a fine-tuned stopping condition that satisfies the
tight and complete requirements to early terminate the index traversal. 

First, the stopping condition has to guarantee \emph{completeness} (Definition \ref{def:complete}), 
i.e. when the stopping condition $\varphi$ holds on a position $\mb$,
the candidate set $\mathcal{C}$ must contain all the true results. 
Note that since the input of $\varphi$ is the partial observation at $\mb$,
we must guarantee that for all possible databases $\db$ consistent with the partial observation $\calL(\mb)$,
the candidate set $\mathcal{C}$ contains all vectors in $\db$ that are $\theta$-similar to the query $\mq$. 
This is equivalent to require that if a unit vector $\ms$ is found below position $\mb$ (i.e. $\ms$ does not appear above $\mb$),
then $\ms$ is NOT $\theta$-similar to $\mq$. 
We formulate this as follows.
\begin{definition}[Completeness]\label{def:complete}
Given a query $\mq$ with threshold $\theta$,
a position vector $\mb$ on index $\{L_i\}_{1 \leq i \leq d}$ is complete
iff for every unit vector $\mathbf{s}$, $\ms < L[\mb]$ implies $\ms \cdot \mq < \theta$.
A stopping condition $\varphi(\cdot)$ is complete iff for every $\mb$,
$\varphi(\mb) = \mathtt{True}$ implies that $\mb$ is complete.
\end{definition}

The second requirement of the stopping condition is \emph{tightness}.
It is desirable that the algorithm terminates immediately once the candidate set $\mathcal{C}$ contains a complete set of candidates, such that no additional 
unnecessary access is made. This can reduce not only the number of index accesses but also the candidate set size, which in turn reduces the verification cost.
Formally,
\begin{definition}[Tightness]\label{def:tight}
A stopping condition $\varphi(\cdot)$ is tight iff for every complete position vector $\mb$, $\varphi(\mb) = \mathtt{True}$.
\end{definition}

It is desirable that a stopping condition be both complete and tight. 
However, as we shown in Appendix~\ref{sec:TANotTight}, the baseline stopping condition $\varphi_{\BL} = \big(\mq \cdot L[\mb] < \theta \big)$
is complete but not tight as it does not capture the unit vector constraint to terminate as soon as no unseen unit vector
can satisfy $\ms \cdot \mq \geq \theta$.
Next, we present a new stopping condition that is both complete and tight. 

To guarantee tightness, one can check at every snapshot during the traversal
whether the current position vector $\mb$ is complete and stop once the condition is true.
However, directly testing the completeness is impractical since
it is equivalent to testing 
whether there exists a real vector $\ms = (s_1, \dots, s_d)$ that satisfies the following 
following set of quadratic constraints:
\vspace{-2mm}
\begin{equation} \label{equ:quadratic}
(a) \quad \sum_{i=1}^d s_i \cdot q_i \geq \theta , \quad \quad
(b) \quad s_i \leq L_i[b_i], \ \forall \ i \in [d], \quad \text{ and} \quad \quad
(c) \quad \sum_{i=1}^d s_i^2 = 1.
\vspace{-1.5mm}
\end{equation}
We denote by $\mC(\mb)$ (or simply $\mC$) the set of $\mathbb{R}^d$ points defined by the above constraints.
The set $\mC(\mb)$ is infeasible (i.e. there is no satisfying $\ms$) if and only if $\mb$ is complete, 
but directly testing the feasibility of $\mC(\mb)$ requires an expensive call to a quadratic programming solver. 
Depending on the implementation, the running time can be exponential or of high-degree polynomial \cite{convexopt}.
We address this challenge by deriving an equivalently strong stopping condition that guarantees tightness and is efficiently testable:
\begin{theorem}\label{thm:qiLibi}
Let $\tau$ be the solution of the equation
$ \sum_{i=1}^{d} \min\{q_i \cdot \tau, L_i[b_i]\} ^ 2 = 1$ and
\begin{equation} \label{equ:score}
\score(L[\mb]) = \sum_{i=1}^{d} \min\{q_i \cdot \tau, L_i[b_i]\}  \cdot q_i
\end{equation}
called the \emph{max-similarity}.
The stopping condition $\varphi_{\TC}(\mb) = \left( \score(L[\mb]) < \theta \right)$ is tight and complete.
\end{theorem} 

\begin{proof}
The tight and complete stopping condition is obtained by applying the Karush-Kuhn-Tucker (KKT) conditions \cite{kkt}
for solving nonlinear programs. We first formulate the set of constraints in (\ref{equ:quadratic}) as an optimization problem
over $\ms$:
\begin{equation} \label{equ:maximization}
\vspace{-1.5mm}
\begin{array}{ll@{}ll@{}ll@{}ll}
\text{maximize}  & \displaystyle\sum_{i=1}^{d} s_i \cdot q_{i} \quad\quad\quad\quad & \text{subject to}& \displaystyle\sum_{i=1}^d s_i^2 = 1 \quad \text{ and } \quad s_i \leq L_i[b_i],                    \quad \forall i \in [d]
\end{array}
\vspace{-1.5mm}
\end{equation}
So checking whether $\mC$ is feasible is equivalent to verifying whether the maximal $\sum_{i=1}^{d} s_i \cdot q_{i}$ is at least $\theta$.
So it is sufficient to show that $\sum_{i=1}^{d} s_i \cdot q_{i}$ is maximized when $s_i = \min\{q_i \cdot \tau, L_i[b_i]\}$ as specified above.

The KKT conditions of the above maximization problem specify a set of necessary conditions that the optimal $\ms$ needs to satisfy.
More precisely, let 
$$
 L(\ms, \mathbf{\mu}, \lambda) = \sum_{i=1}^d s_i q_i - \sum_{i=1}^d \mu_i(L_i[b_i] - s_i) - \lambda \left(\sum_{i=1}^d s_i^2 - 1\right) 
$$
be the Lagrangian of (\ref{equ:maximization}) where $\lambda \in \mathbb{R}$ and $\mathbf{\mu} \in \mathbb{R}^d$ are the Lagrange multipliers.
Then,
\begin{lemma}[derived from KKT]
The optimal $\ms$ in (\ref{equ:maximization}) satisfies the following conditions:
\begin{equation*}
\begin{array}{ll@{}ll}
\nabla_{\ms} L(\ms, \mu, \lambda) = 0 & \text{(Stationarity)} \\
\mu_i \geq 0 , \ \forall \ i \in [d] & \text{(Dual feasibility)} \\
\mu_i (L_i[b_i] - s_i) = 0, \ \forall \ i \in [d] & \text{(Complementary slackness)}
\end{array}
\end{equation*}
in addition to the constraints in (\ref{equ:maximization}) (called the \emph{Primal feasibility} conditions).
\end{lemma}
By the Complementary slackness condition, for every $i$, if $\mu_i \neq 0$ then
$s_i = L_i[b_i]$. If $\mu_i = 0$, then from the Stationarity condition, 
we know that for every $i$, $ q_i + \mu_i - 2 \lambda \cdot s_i = 0 $
so $s_i = q_i / 2\lambda$. Thus, the value of $s_i$ is either $L_i[b_i]$ or $q_i / 2\lambda$.

If $L_i[b_i] < q_i / 2\lambda$ then since $s_i \leq L_i[b_i]$, the only possible case is $s_i = L_i[b_i]$.
For the remaining dimensions, the objective function $\sum_{i=1}^d s_i \cdot q_i$ is maximized
when each $s_i$ is proportional to $q_i$, so $s_i = q_i / 2\lambda$.
Combining these two cases, we have $s_i = \min\{q_i / 2\lambda, L_i[b_i]\}.$

Thus, for the $\lambda$ that satisfies 
$\sum_{i=1}^d \min\{q_i / 2\lambda, L_i[b_i]\}^2 = 1$,
the objective function $\sum_{i=1}^d s_i \cdot q_i$ is maximized when
$s_i = \min\{q_i / 2\lambda, L_i[b_i]\}$ for every $i$.
The theorem is obtained by letting $\tau = 1/2\lambda$.
\end{proof}
\vspace{-1mm}

\noindent
\textbf{Remark of $\varphi_{\TC}$. } 
The tight stopping condition $\varphi_{\TC}$ computes
the vector $\ms$ below $L(\mb)$ with the maximum cosine similarity $\score(L[\mb])$ with the query $\mq$.
At the beginning of the gathering phase, $b_i = 0$ for every $i$ so 
$\score(L[\mb]) = 1$ as $\ms$ is not constrained. 
The cosine score is maximized when $\ms = \mq$ where $\tau = 1$.
During the gathering phase, as $b_i$ increases, the upper bound $L_i[b_i]$ of each $s_i$ decreases.
When $L_i[b_i] < q_i$ for some $i$, $s_i$ can no longer be $q_i$. 
Instead, $s_i$ equals $L_i[b_i]$, the rest of $\ms$ increases proportional to $\mq$ and $\tau$ increases.
During the traversal, the value of $\tau$ monotonically increases and the score $\ms(L[\mb])$ monotonically decreases.
This is because the space for $\ms$ becomes more constrained by $L(\mb)$ as the pointers move deeper in the inverted lists.

Testing the tight and complete condition $\varphi_{\TC}$ requires solving $\tau$ in Theorem (\ref{thm:qiLibi}), 
for which a direct application of the bisection method takes $\calO(d)$ time.
We show a novel efficient algorithm (Appendix \ref{sec:incremental}) based on incremental maintenance which takes only
$\calO(\log d)$ time for each test of $\varphi_{\TC}$. 
\vspace{-2mm}
\begin{theorem}\label{thm:incremental}
The stopping condition $\varphi_{\TC}(\mb)$ can be incrementally computed in $\calO(\log d)$ time.
\end{theorem}
\vspace{-2mm}

\vspace{-1mm}
\section{Near-Optimal Traversal Strategy} \label{sec:hull-based}


Given the inverted lists index and a query, there can be many stopping positions
that are both complete and tight. To optimize the performance, we need a traversal strategy that
reaches one such position as fast as possible.
Specifically, the goal is to design a traversal strategy $\calT$ that minimizes $|\mb| = \sum_{i=1}^d b_i$
where $\mb$ is the first position vector satisfying the tight and complete stopping condition if $\calT$ is followed.
Minimizing $|\mb|$ also reduces the number of collected candidates,
which in turn reduces the cost of the verification phase.
We call $|\mb|$ the \emph{access cost} of the strategy $\calT$.
Formally,
\vspace{-1mm}
\begin{definition}[Access Cost]\label{def:optimal}
Given a traversal strategy $\calT$, we denote by $\{\mb_i\}_{i \geq 0}$
the sequence of position vectors obtained by following $\calT$.
The access cost of $\calT$, denoted by $\cost(\calT)$, is the minimal $k$ such that
$\varphi_{\TC}(\mb_k) = \mathtt{True}$. Note that $\cost(\calT)$ also equals $|\mb_k|$.
\end{definition}

\vspace{-1mm}
\begin{definition}[Instance Optimality]\label{def:optimal}
Given a database $\db$ with inverted lists $\{L_i\}_{1 \leq i \leq d}$, a query vector $\mq$ and a threshold $\theta$,
the optimal access cost $\opt(\db, \mq, \theta)$ is the minimum $\sum_{i=1}^d b_i$ for position vectors
$\mb$ such that $\varphi_{\TC}(\mb) = \mathtt{True}$.
When it is clear from the context, we simply denote $\opt(\db, \mq, \theta)$ as $\opt(\theta)$ or $\opt$.
\end{definition}
\vspace{-1mm}

At a position $\mb$, a traversal strategy makes its decision locally
based on what has been observed in the inverted lists up to that point,
so the capability of making globally optimal decisions is limited.
As a result, traversal strategies are often designed as simple heuristics,
such as the lockstep strategy in the baseline approach.
The lockstep strategy has a $d \cdot \opt$ near-optimal bound which is loose in the high-dimensionality setting.

In this section, we present a traversal strategy for cosine threshold queries
with tighter near-optimal bound by taking into account that
the index values are skewed in many realistic scenarios. 
We approach the (near-)optimal traversal strategy in two steps. 

First, we consider the simplified case with the unit-vector constraint ignored
so that the problem is reduced to inner product queries.
We propose a general traversal strategy that relies on convex hulls pre-computed from the inverted lists during indexing.
During the gathering phase, these convex hulls are accessed as auxiliary data during the traversal
to provide information on the increase/decrease rate towards the stopping condition.
The hull-based traversal strategy not only makes fast decisions (in $\calO(\log d)$ time)
but is near-optimal (Corollary \ref{cor:near-optimal}) under a reasonable assumption.
In particular, we show that if the distance between any two consecutive convex hull vertices of the inverted lists is bounded by a constant $c$,
the access cost of the strategy is at most $\opt + c$.
Experiments on real data show that this constant is small in practice.

The hull-based traversal strategy is quite general, as it applies to a large class of functions beyond inner product
called the \emph{decomposable functions}, which have the form $\sum_{i=1}^d f_i(s_i)$
where each $f_i$ is a non-decreasing real function of a single dimension $s_i$.
Obviously, for a fixed query $\mq$, the inner product $\mq \cdot \ms$ is a special case of decomposable functions,
where each $f_i(s_i) = q_i \cdot s_i$. We show that the near-optimality result for inner product queries
can be generalized to any decomposable function (Theorem \ref{thm:near-optimal}).

Next, in Section \ref{sec:global-optimal}, we consider the cosine queries by taking the normalization constraint into account.
Although the function $\score(\cdot)$ used in the tight stopping condition $\varphi_{\TC}$ is not decomposable so
the same technique cannot be directly applied,
we show that the hull-based strategy can be adapted by approximating
$\score(\cdot)$ with a decomposable function.
In addition, we show that with a properly chosen approximation,
the hull-based strategy is near-optimal with a small adjustment to the input threshold $\theta$,
meaning that the access cost is bounded by $\opt(\theta - \epsilon) + c$ for a small $\epsilon$ (Theorem \ref{thm:near-optimal-approx}).
Under the same experimental setting, we verify that $\epsilon$ is indeed small in practice.

\vspace{-1mm}
\subsection{Decomposable Functions}

We start with defining the decomposable functions for which the hull-based traversal strategies can be applied:
\vspace{-1mm}
\begin{definition}[Decomposable Function]\label{def:linear-score}
A decomposable function $F(\ms)$ is a $d$-dimensional real function where
$
F(\ms) = \sum_{i=1}^d f_i(s_i)
$
and each $f_i$ is a non-decreasing real function.
\end{definition}
\vspace{-1mm}


Given a decomposable function $F$, the corresponding stopping condition is called a \emph{decomposable condition}, which we define next.
\vspace{-1mm}
\begin{definition}[Decomposable Condition]\label{def:linear}
A decomposable condition $\varphi_F$ is a boolean function
$\varphi_F(\mb) = \big(F(L[\mb])  < \theta \big)$
where $F$ is a decomposable function and $\theta$ is a fixed threshold.
\end{definition}
\vspace{-1mm}

When the unit vector constraint is lifted, the decomposable condition is tight and complete for any scoring function $F$ and threshold $\theta$.
As a result, the goal of designing a traversal strategy for $F$ is to have the access cost
as close as possible to $\opt$ when the stopping condition is $\varphi_F$.


\vspace{-1mm}
\subsection{The max-reduction traversal strategy}


To illustrate the high-level idea of the hull-based approach,
we start with a simple greedy traversal strategy called the \emph{Max-Reduction} traversal strategy $\calT_{\mathsf{MR}}(\cdot)$.
The strategy works as follows: at each snapshot, move the pointer $b_i$ on the inverted list $L_i$ that results in the maximal reduction on the score
$F(L[\mb])$. Formally, we define
$$
\calT_{\mathsf{MR}}(\mb) = \argmax_{1 \leq i \leq d} \left( F(L[\mb]) - F(L[\mb + \me_i]) \right) = \argmax_{1 \leq i \leq d} \left( f_i(L_i[b_i]) - f_i(L_i[b_i + 1]) \right)
$$
where $\me_i$ is the vector with 1 at dimension $i$ and 0's else where.
Such a strategy is reasonable since one would like $F(L[\mb])$ to drop as fast as possible,
so that once it is below $\theta$, the stopping condition $\varphi_F$ will be triggered and terminate the traversal.

It is obvious that there are instances where the max-reduction strategy can be far from optimal,
but is it possible that it is optimal under some assumption? The answer is positive:
if for every list $L_i$, the values of $f_i(L_i[b_i])$ are decreasing at decelerating rate,
then we can prove that its access cost is optimal.
We state this ideal assumption next.
\vspace{-1mm}
\begin{assumption}[Ideal Convexity]\label{ass:delta}
For every inverted list $L_i$, let $\Delta_i[j] = f_i(L_i[j]) - f_i(L_i[j + 1])$ for $0 \leq j < |L_i|$.\footnote{Recall that $L_i[0] = 1$.}
The list $L_i$ is ideally convex if the sequence $\Delta_i$ is non-increasing, i.e., $\Delta_i[j+1] \le \Delta_i[j]$ for every $j$.
Equivalently, the piecewise linear function 
passing through the points $\{(j, f_i(L_i[j]))\}_{0 \leq j \leq |L_i|}$ is convex for each $i$.
A database $\db$ is ideally convex if every $L_i$ is ideally convex.
\end{assumption}
\vspace{-1.5mm}

An example of an inverted list satisfying the above assumption is shown in Figure \ref{fig:convex}(a).
The max-reduction strategy $\calT_{\mathsf{MR}}$ is optimal under the ideal convexity assumption:
\begin{theorem}[Ideal Optimality]\label{thm:optimal}
Given a decomposable function $F$, for every ideally convex database $\db$ and every threshold $\theta$,
the access cost of $\calT_{\mathsf{MR}}$ is exactly $\opt$.
\end{theorem}
We prove Theorem \ref{thm:optimal} with a simple greedy argument (Appendix \ref{thm:proof-optimal}):
each move of $\calT_{\mathsf{MR}}$ always results in the globally maximal reduction in the scoring function
as guaranteed by the convexity condition. \yuliang{moved the proof to the appendix and replace with the above paragraph. Please check.}



\vspace{-1mm}
\subsection{The hull-based traversal strategy}
\label{sec:near-optimal}

Theorem~\ref{thm:optimal} provides a strong performance guarantee but the ideal convexity assumption is usually not true on real datasets.
Without the ideal convexity assumption, the strategy suffers from the drawback of making locally optimal but globally suboptimal decisions.
The pointer $b_i$ to an inverted list $L_i$ might never be moved if choosing the current $b_i$ only results in a small
decrease in the score $F(L[\mb])$, but there is a much larger decrease several steps ahead.
As a result, the $\calT_{\mathsf{MR}}$ strategy has no performance guarantee in general.

\begin{figure}[!t]
\centering
\includegraphics[width=.7\columnwidth]{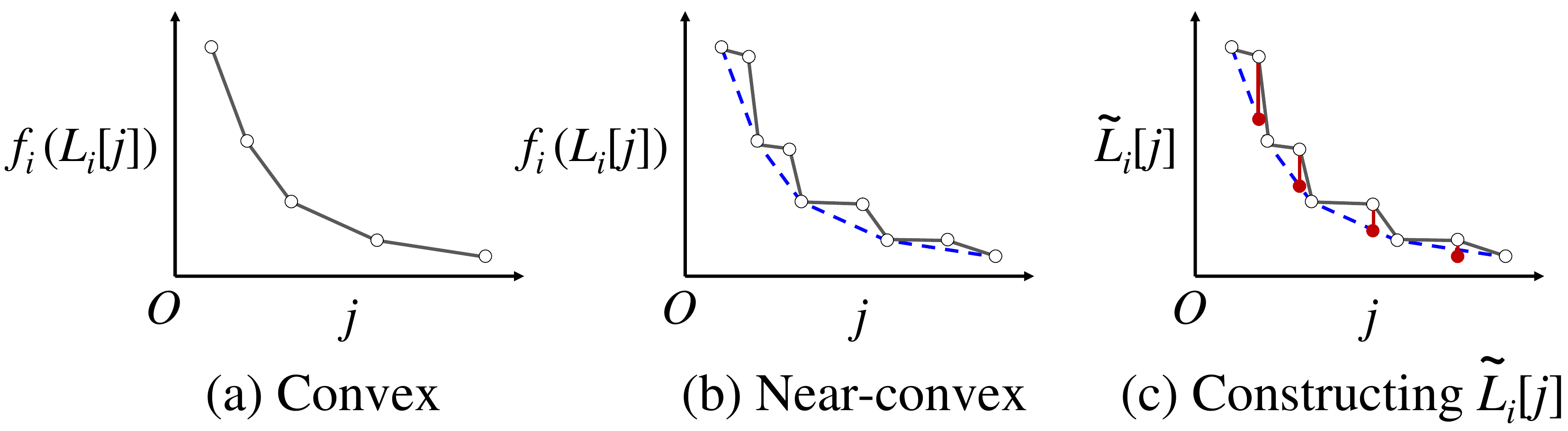}
\caption{{Convexity and near-convexity}}\label{fig:convex}
\end{figure}

\begin{figure}[!t]
\renewcommand{\tabcolsep}{0.1mm}
\begin{tabular}{ccccc}
\includegraphics[width=0.25\columnwidth]{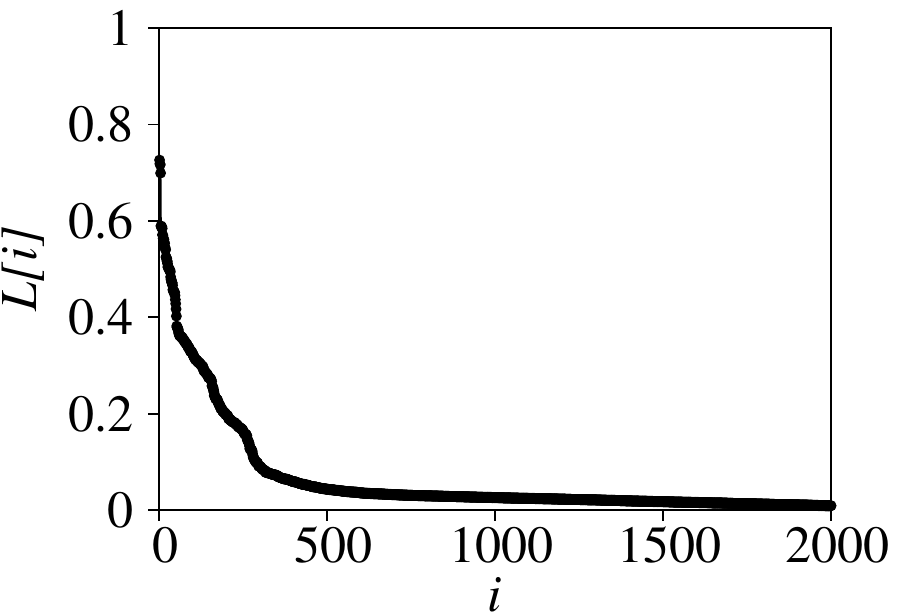}&\includegraphics[width=0.25\columnwidth]{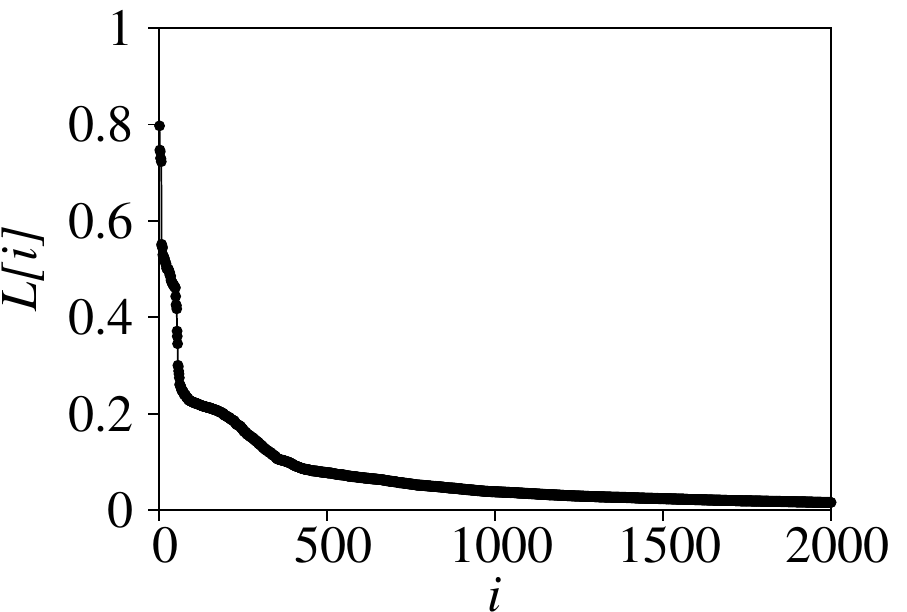}&
\includegraphics[width=0.25\columnwidth]{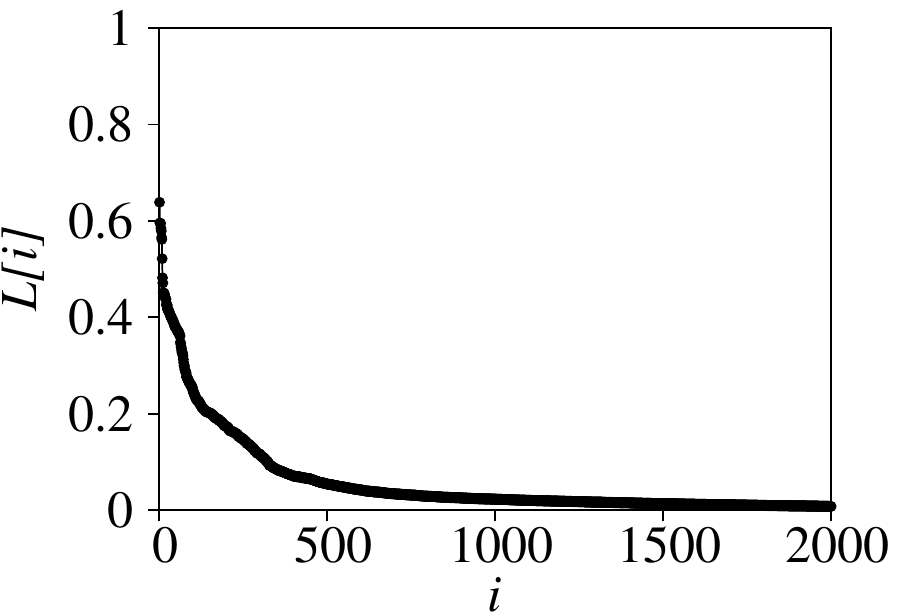}&\includegraphics[width=0.25\columnwidth]{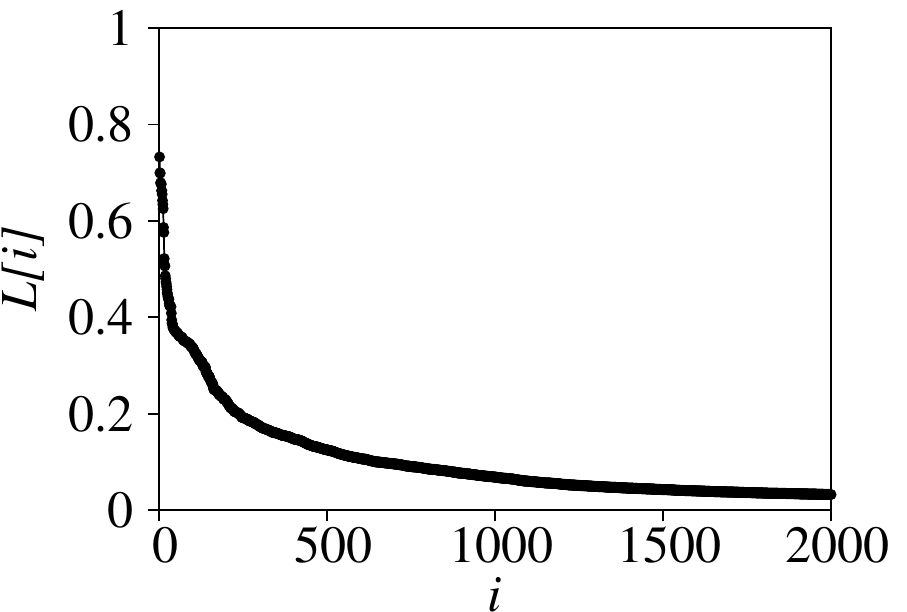}\\
(a) list $L_0$ & (b) list $L_1$ & (c) list $L_2$ & (d) list $L_3$
\end{tabular}
\caption{The skewed inverted lists in mass spectrometry (See Figure \ref{fig:near-convex-doc} and \ref{fig:near-convex-img}
for similar patterns observed in document and image datasets)}\label{fig:inverted-list}
\vspace{-2mm}
\end{figure}

In most practical scenarios that we have seen, 
we can bring the traversal strategy $\calT_{\MR}$ to practicality by considering a relaxed version of Assumption \ref{ass:delta}.
Informally, instead of assuming that each list $f_i(L_i)$ forms a convex piecewise linear function,
we assume that $f_i(L_i)$ is ``mostly'' convex, meaning that if we compute the \emph{lower convex hull} \cite{de2008computational} of $f_i(L_i)$,
the gap between any two consecutive vertices on the convex hull is small.\footnote{We denote by $f_i(L_i)$ the list $[f_i(L_i[0]), f_i(L_i[1]), \dots]$ for every $L_i$.}
Intuitively, the relaxed assumption implies that the values at each list are decreasing at ``approximately'' decelerating speed.
It allows list segments that do not follow the overall deceleration trend, as long as their lengths are bounded by a constant.
We verified this property in the mass spectrometry dataset as illustrated in Figure \ref{fig:inverted-list},
a document dataset, and an image dataset (Figure \ref{fig:near-convex-doc} and \ref{fig:near-convex-img} in Appendix \ref{sec:more-experiment}).
\vspace{-1mm}
\begin{assumption}[Near-Convexity]\label{ass:delta-approx}
For every inverted list $L_i$, let $\hull_i$ be the lower convex hull of the set
of 2-D points $\{(j, f_i(L_i[j]))\}_{0 \leq j \leq |L_i|}$ represented by a set of indices
$\hull_i = \{j_1, \dots, j_n\}$ 
where for each $1 \leq k \leq n$, $(j_k, f_i(L_i[j_k]))$ is a vertex of the convex hull.
The list $L_i$ is near-convex if for every $k$, $j_{k + 1} - j_k$ is upper-bounded by some constant $c$.
A database $\db$ is near-convex if every inverted list $L_i$ is near-convex with the same constant $c$,
which we refer to as the convexity constant.
\end{assumption}

\vspace{-1mm}
\begin{example}\label{exm:convex}
Intuitively, the near-convexity assumption captures the case where each $f_i(L_i)$
is decreasing with \emph{approximately} decelerating speed, so the number of points between two convex hull vertices should be small.
For example, when $f_i$ is a linear function, the list $L_i$ shown in Figure \ref{fig:convex}(b)
is near-convex with convexity constant 2 since there is at most 1 point between each pair of consecutive vertices of the convex hull (dotted line).
In the ideal case shown in Figure \ref{fig:convex}(a), the constant is 1
when the decrease between successive values is strictly decelerating. 
\end{example}
\vspace{-1mm}

Imitating the max-reduction strategy, for every pair of consecutive indices $j_k, j_{k+1}$ in $\hull_i$ and
for every index $j \in [j_k, j_{k+1})$, let $\tilde{\Delta}_i[j] = \dfrac{f_i(L_i[j_k]) - f_i(L_i[j_{k + 1}])}{j_{k + 1} - j_k}$.
Since the $(j_k, f_i(L_i[j_k]))$'s are vertices of a lower convex hull,
each sequence $\tilde{\Delta}_i$ is non-decreasing.
Then the \emph{hull-based} traversal strategy is simply defined as
\begin{equation}
\calT_{\HL}(\mb) = \argmax_{1 \leq i \leq d}(\tilde{\Delta}_i[b_i]).
\end{equation}
\noindent
\textbf{Remark on data structures}.
In a practical implementation, to answer queries with scoring function $F$ using the hull-based strategy,
the lower convex hulls need to be ready before the traversal starts.
If $F$ is a general function unknown a priori, the convex hulls need to be computed online which is not practical.
Fortunately, when $F$ is the inner product $F(\ms) = \mq \cdot \ms$ parameterized by the query $\mq$,
each convex hull $\hull_i$ is exactly the convex hull for the points $\{(j, L_i[j])\}_{0 \leq i \leq |L_i|}$ from $L_i$.
This is because the slope from any two points $(j, f_i(L_i[j]))$ and $(k, f_i(L_i[k]))$
is $\dfrac{q_i L_i[j] - q_i L_i[k]}{j - k} $, which is exactly the slope from $(j, L_i[j])$ and $(k, L_i[k])$
multiplied by $q_i$.
So by using the standard convex hull algorithm \cite{de2008computational},
$\hull_i$ can be pre-computed in $\calO(|L_i|)$ time.
Then the set of the convex hull vertices $\hull_i$ can be stored as inverted lists
and accessed for computing the $\tilde{\Delta}_i$'s during query processing.
In the ideal case, $\hull_i$ can be as large as $|L_i|$ but is much smaller in practice.

Moreover, during the traversal using the strategy $\calT_{\HL}$,
choosing the maximum $\tilde{\Delta}_i[b_i] $ at each step can be done in $\calO(\log d)$ time using a max heap.
This satisfies the requirement that the traversal strategy is efficiently computable.

%

\smallskip
\noindent
\textbf{Near-optimality results. }
We show that the hull-based strategy $\calT_{\HL}$ is 
near-optimal under the near-convexity assumption. 
\vspace{-1.5mm}
\begin{theorem}\label{thm:near-optimal}
Given a decomposable function $F$, for every near-convex database $\db$ and every threshold $\theta$,
the access cost of $\calT_{\HL}$ is strictly less than
$\opt + c$ where $c$ is the convexity constant.
\vspace{-1.5mm}
\end{theorem}
When the assumption holds with a small convexity constant,
this near-optimality result provides a much tighter bound 
compared to the $d \cdot \opt$ bound in the $\TA$-inspired baseline.
This is achieved under data assumption and by keeping the convex hulls
as auxiliary data structure, so it does not contradict 
the lower bound results on the approximation ratio \cite{Fagin2001OAA}.

\vspace{-1mm}
\begin{proof}
Let $\calB = \{\mb_i\}_{i \geq 0}$ be the sequence of position vectors generated by $\calT_{\HL}$.
We call a position vector $\mb$ a \emph{boundary position} if every $b_i$ is the index of a vertex of the convex hull $\hull_i$.
Namely, $b_i \in \hull_i$ for every $i \in [d]$.
Notice that if we break ties consistently during the traversal of $\calT_{\HL}$,
then in between every pair of consecutive boundary positions $\mb$ and $\mb'$ in $\calB$,
$\calT_{\HL}(\mb)$ will always be the same index.
We call the subsequence positions $\{\mb_i\}_{l \leq i < r}$ of $\calB$ where $\mb_l = \mb$ and $\mb_r = \mb'$
a \emph{segment} with boundaries $(\mb_l, \mb_r)$.
We show the following lemma.
\vspace{-1mm}
\begin{lemma}\label{lem:bstar}
For every boundary position vector $\mb$ generated by $\calT_{\HL}$,
we have $F(L[\mb]) \leq F(L[\mb^*])$ for every position vector $\mb^*$ where $|\mb^*| = |\mb|$.
\vspace{-1mm}
\end{lemma}
Intuitively, the above lemma says that if the traversal of $\calT_{\HL}$ reaches a boundary position $\mb$,
then the score $F(L[\mb])$ is the minimal possible score obtained by any traversal sequence of at most $|\mb|$ steps.
We prove Lemma \ref{lem:bstar} by generalizing the greedy argument in the proof of Theorem \ref{thm:optimal}.
More details can be found in Appendix \ref{sec:proof-near-optimal}.

Lemma \ref{lem:bstar} is sufficient for Theorem \ref{thm:near-optimal} because of the following.
Suppose $\mb_{\stp}$ is the stopping position in $\calB$, which means that
$\mb_{\stp}$ is the first position in $\calB$ that satisfies $\varphi_F$ and the access cost is $|\mb_{\stp}|$.
Let $\{\mb_i\}_{l \leq i < r}$ be the segment that contains $\mb_{\stp}$.
Given Lemma \ref{lem:bstar}, Theorem \ref{thm:near-optimal} holds trivially if $\mb_{\stp} = \mb_l$.
It remains to consider the case $\mb_{\stp} \neq \mb_l$. Since the traversal does not stop at $\mb_l$, we have $F(L[\mb_l]) \geq \theta$.
By Lemma \ref{lem:bstar}, $\mb_l$ is the position with minimal $F(L[\cdot])$ obtained in $|\mb_l|$ steps
so $|\mb_l| \leq \opt$.
Since $|\mb_\stp| - |\mb_l| < |\mb_r| - |\mb_l| \leq c$, we have that $|\mb_\stp| < \opt + c$.
We illustrate this in Figure \ref{fig:traversalExample}.
\end{proof}
\vspace{-1mm}

\begin{figure}[!t]
\centering
\renewcommand{\tabcolsep}{0.1mm}
\includegraphics[width=0.48\columnwidth]{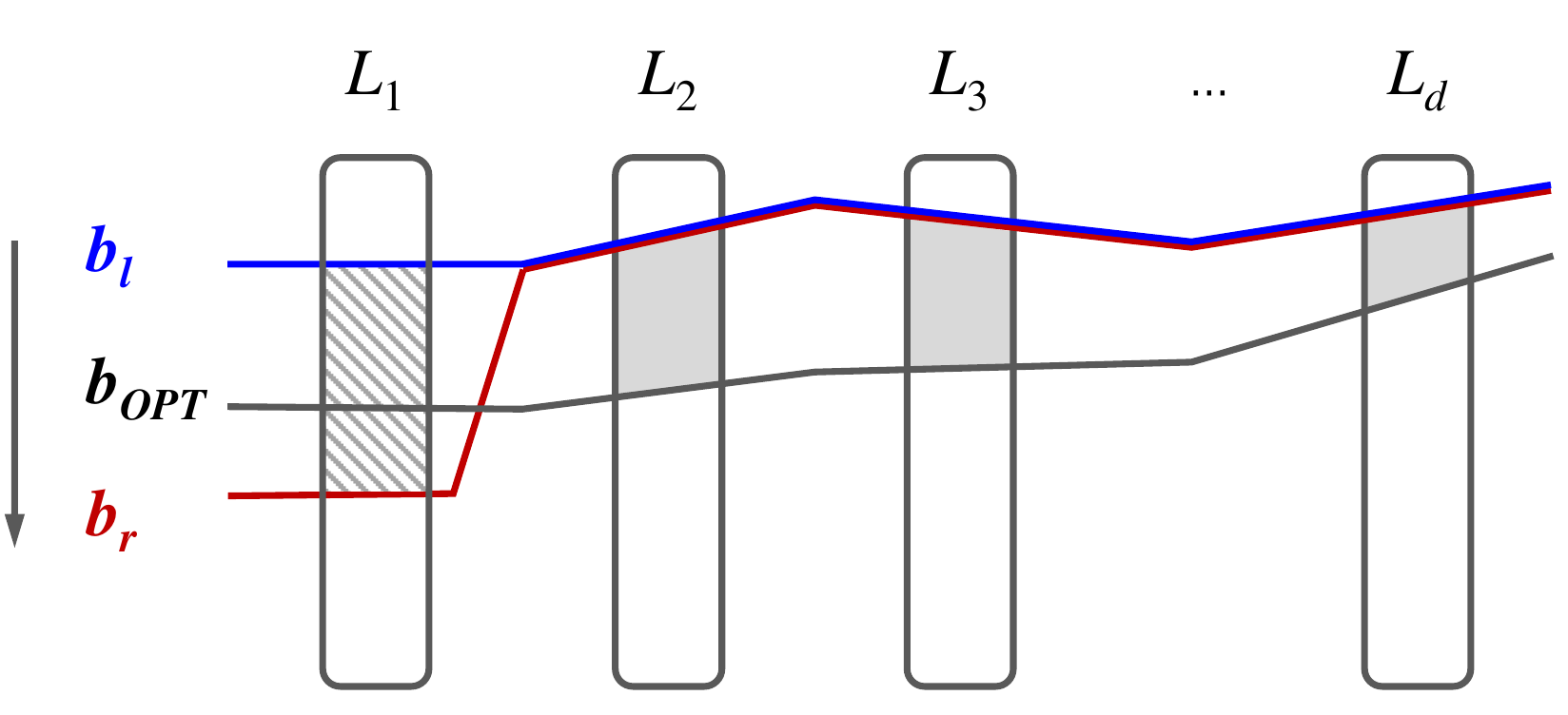}
\caption{($\mb_l$, $\mb_r$): the two boundary positions surrounding the stopping position $\mb_\stp$ of $\calT_{\HL}$; $\mb_{\opt}$: the optimal stopping position; It is guaranteed that
(1) $|\mb_\stp| - |\mb_l| < |\mb_r| - |\mb_l| \leq c$ and
(2) $|\mb_l| < |\mb_\opt|$.}\label{fig:traversalExample}
\vspace{-1mm}
\end{figure}

Since the baseline stopping condition $\varphi_{\BL}$ is tight and complete for inner product queries,
one immediate implication of Theorem \ref{thm:near-optimal} is that
\vspace{-1mm}
\begin{corollary}\label{cor:near-optimal}
(Informal) The hull-based strategy $\calT_{\HL}$ for inner product queries is near-optimal.
\end{corollary}
\vspace{-1mm}

\smallskip
\noindent
\textbf{Verifying the assumption}.
We demonstrate the practical impact of the near-optimality result in real mass spectrometry datasets. 
The same experiment is repeated on a document and an image dataset in Appendix \ref{sec:more-experiment}.
The near-convexity assumption requires that the gap between any two consecutive convex hull
vertices has bounded size, which is hard to achieve in general.
According to the proof of Theorem \ref{thm:near-optimal}, for a given query,
the difference from the optimal access cost is at most the size of the gap
between the two consecutive convex hull vertices containing the last move of the strategy (the $\mb_l$ and $\mb_r$ in Figure \ref{fig:traversalExample}).
The size of this gap can be much smaller than the global convexity constant $c$,
so the overall precision can be much better in practice.
We verify this by running a set of 1,000 real queries on the dataset\footnote{\url{https://proteomics2.ucsd.edu/ProteoSAFe/index.jsp}}.
\yuliang{To Jianguo: can we move the link to the reference? I think it helps saving some space.}
The gap size is 163.04 in average, which takes only 1.3\% of the overall access cost of traversing the indices.
This indicates that the near-optimality guarantee holds in the mass spectrometry dataset.
Similar results are obtained in a document and an image dataset, where the
gap size takes only 7.9\% and 0.4\% of the overall access cost respectively.

\vspace{-1mm}
\subsection{The traversal strategy for cosine}\label{sec:global-optimal}
\vspace{-1mm}

Next, we consider traversal strategies which take into account the unit vector constraint posed by the cosine function,
which means that the tight and complete stopping condition is $\varphi_{\TC}$ introduced in Section \ref{sec:stop}.
However, since the scoring function $\score$ in $\varphi_{\TC}$ is not decomposable,
the hull-based technique cannot be directly applied.
We adapt the technique by approximating the original $\score$ with a decomposable function $\tilde{F}$.
Without changing the stopping condition $\varphi_{\TC}$,
the hull-based strategy can then be applied with the convex hull indices constructed with the approximation $\tilde{F}$.
In the rest of this section, we first generalize the result in Theorem \ref{thm:near-optimal} to scoring functions having
decomposable approximations and show how the hull-based traversal strategy can be adapted.
Next, we show a natural choice of the approximation for $\score$ with practically tight near-optimal bounds.
Finally, we discuss data structures to support fast query processing using the traversal strategy.

We start with some additional definitions.
\vspace{-1mm}
\begin{definition}
A $d$-dimensional function $F$ is decomposably approximable if there exists
a decomposable function $\tilde{F}$, called the \emph{decomposable approximation} of $F$,
and two non-negative constants $\epsilon_1$ and $\epsilon_2$
such that $\tilde{F}(\ms) - F(\ms) \in [-\epsilon_1, \epsilon_2]$ for every vector $\ms$.
\vspace{-1mm}
\end{definition}

When applied to a decomposably approximable function $F$,
the hull-based traversal strategy $\calT_{\HL}$ is adapted by constructing the
convex hull indices and the $\{\tilde{\Delta}_i\}_{1 \leq i \leq d}$
using the approximation $\tilde{F}$. The following can be obtained by generalizing Theorem \ref{thm:near-optimal}:
\vspace{-1mm}
\begin{theorem}\label{thm:near-optimal-approx}
Given a function $F$ approximable by a decomposable function $\tilde{F}$ with constants $(\epsilon_1, \epsilon_2)$,
for every near-convex database $\db$ wrt $\tilde{F}$ and every threshold $\theta$,
the access cost of $\calT_{\HL}$ is strictly less than
$\opt(\theta - \epsilon_1 - \epsilon_2) + c$ where $c$ is the convexity constant.
\vspace{-1mm}
\end{theorem}
\begin{proof}
Recall that $\mb_l$ is the last boundary position generated by $\calT_{\HL}$ that does not satisfy the tight stopping condition for $F$
(which is $\varphi_{\TC}$ when $F$ is $\score$) so $F(L[\mb_l]) \geq \theta$.
It is sufficient to show that for every vector $\mb^*$ where $|\mb^*| = |\mb_l|$, $F(L[\mb^*]) \geq \theta - \epsilon_1 - \epsilon_2$
so no traversal can stop within $|\mb_l|$ steps, implying that the final access cost is no more than $|\mb_l| + c$ which is bounded by
$\opt(\theta - \epsilon_1 - \epsilon_2) + c$.

By Lemma \ref{lem:bstar}, we know that for every such $\mb^*$, $\tilde{F}(L[\mb^*]) \geq \tilde{F}(L[\mb_l])$.
By definition of the approximation $\tilde{F}$, we know that $F(L[\mb^*]) \geq \tilde{F}(L[\mb^*]) - \epsilon_1 $ and
$\tilde{F}(L[\mb_l]) \geq F(L[\mb_l]) - \epsilon_2$.
Combined together, for every $\mb^*$ where $|\mb^*| = |\mb_l|$, we have
\vspace{-1mm}
$$
F(L[\mb^*]) \geq \tilde{F}(L[\mb^*]) - \epsilon_1 \geq \tilde{F}(L[\mb_l]) - \epsilon_1 
            \geq F(L[\mb_l]) - \epsilon_1 - \epsilon_2 \geq \theta - \epsilon_1 -
            \epsilon_2.
\vspace{-1mm}
$$
This completes the proof of Theorem \ref{thm:near-optimal-approx}.
\end{proof}
\smallskip
\noindent
\textbf{Choosing the decomposable approximation. } By Theorem \ref{thm:near-optimal-approx},
it is important to choose an approximation $\tilde{F}$ of $\score$ with small $\epsilon_1$ and $\epsilon_2$
for a tight near-optimality result. By inspecting the formula (\ref{equ:score}) of $\score$, one reasonable choice of $\tilde{F}$
can be obtained by replacing the term $\tau$ with a fixed constant $\tilde{\tau}$. Formally, let
\vspace{-1mm}
\begin{equation}
\tilde{F}(L[\mb]) = \sum_{i=1}^{d} \min\{q_i \cdot \tilde{\tau}, L_i[b_i]\}  \cdot q_i
\vspace{-1mm}
\end{equation}
be the decomposable approximation of $\score$ where each component is a non-decreasing function
$f_i(x) = \min\{q_i \cdot \tilde{\tau}, x\} \cdot q_i$ for $i \in [d]$.

Ideally, the approximation is tight if the constant $\tilde{\tau}$ is close to the final value of $\tau$
which is unknown in advance.
We argue that when $\tilde{\tau}$ is properly chosen,
the approximation parameter $\epsilon_1 + \epsilon_2$ is very small.
With a detailed analysis in Appendix \ref{app:estimation},
we obtain the following upper bound of $\epsilon$:
\vspace{-1mm}
\begin{equation}\label{equ:epsilon_main}
 \epsilon \leq \max\{0, \tilde{\tau} - 1 / \score(L[\mb_l])\} + \score(L[\mb_l]) - \tilde{F}(L[\mb_l]) .
\vspace{-1mm}
\end{equation}

\smallskip
\noindent
\textbf{Verifying the near-optimality. } Next, we verify that the above upper bound of $\epsilon$ is small in practice.
We ran the same set of queries as in Section \ref{sec:near-optimal}
and show the distribution of $\epsilon$'s upper bounds in Figure \ref{fig:epsilon}.
We set $\tilde{\tau} = 1 / \theta$ for all queries so the first term of (\ref{equ:epsilon_main}) becomes zero.
Note that more aggressive pruning can yield better $\epsilon$, but it is not done here for simplicity.
Overall, the fraction of queries with an upper bound $<$0.12 (the sum of the first 3 bars for all $\theta$) is 82.5\%
and the fraction of queries with $\epsilon>0.16$ is $0.5\%$. Similar to the case with inner product queries,
the average of the convexity constant $c$ is 193.39, which is only 4.8\% of the overall access cost.


\begin{figure}[!t]
\centering
\begin{minipage}[t]{0.48\textwidth}
\centering
\includegraphics[width=1.0\columnwidth]{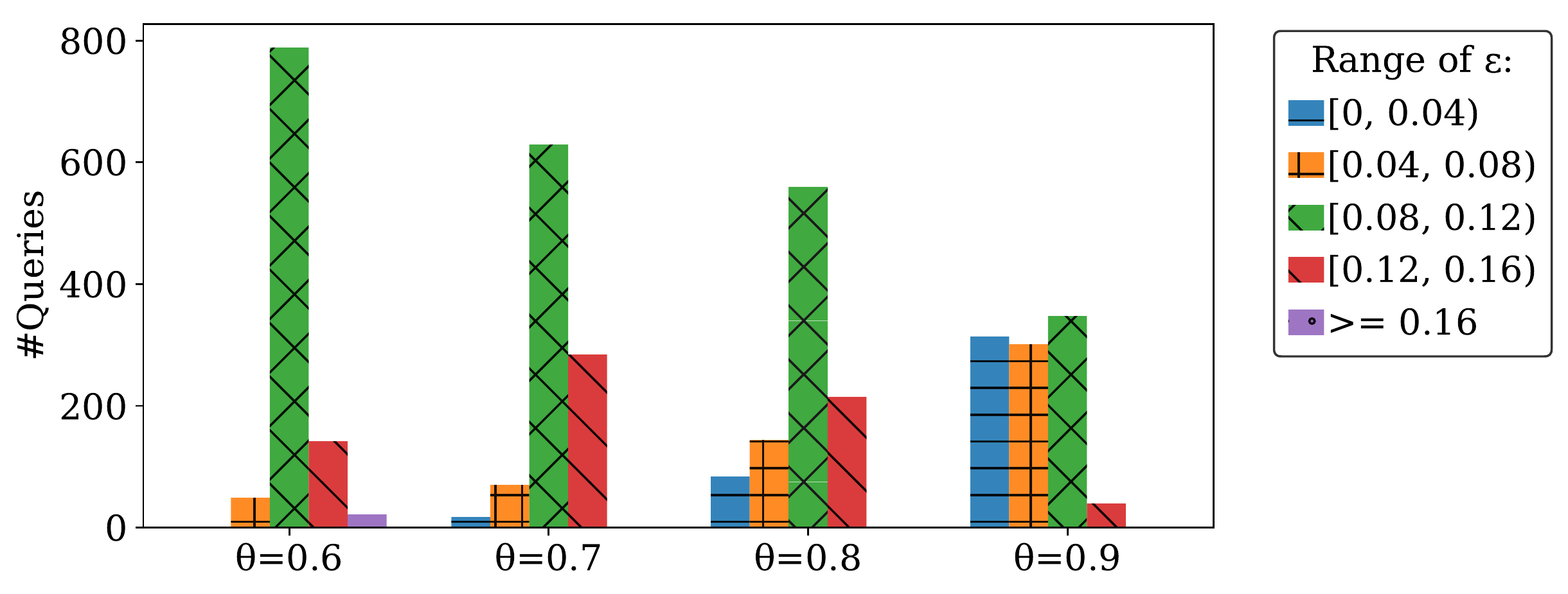}
\caption{{The distribution of $\epsilon$}}\label{fig:epsilon}
\end{minipage}
\begin{minipage}[t]{0.48\textwidth}
\centering
\includegraphics[width=1.0\columnwidth]{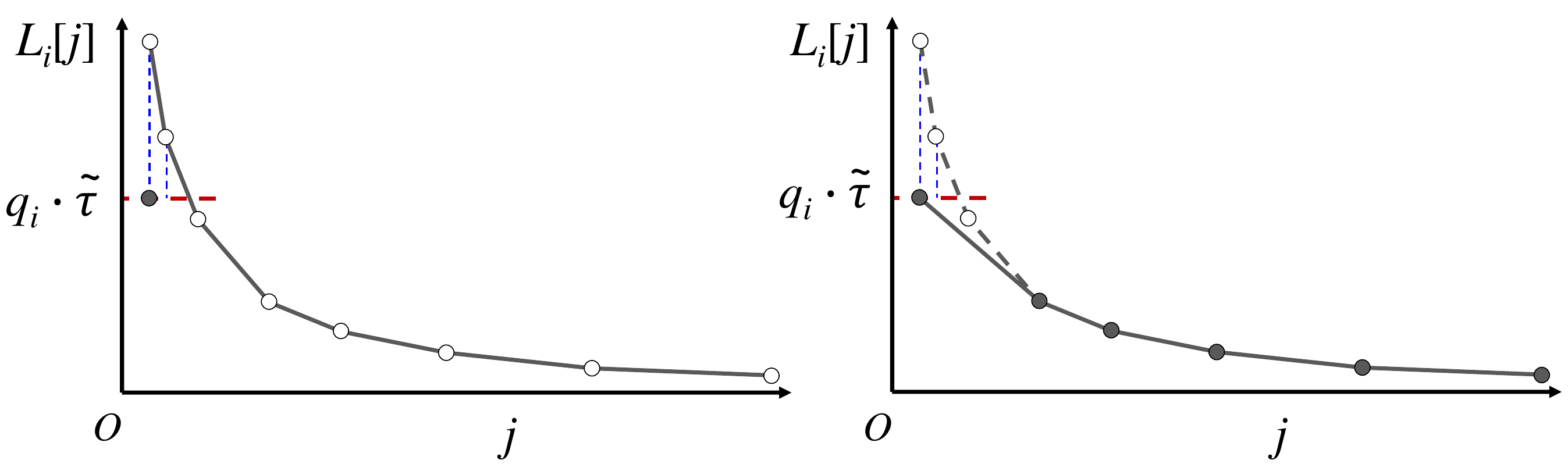}
\caption{{The construction of convex hull $\tilde{\hull}_i$} \yuliang{To Jianguo: can you help enlarge the font size in this figure? THX!}}\label{fig:projection}
\end{minipage}
\end{figure}

\vspace*{1mm}
\noindent
\textbf{Remark on data structures. } Similar to the inner product case,
it is necessary that the convex hulls for $\calT_{\HL}$ can be efficiently obtained without a full computation when a query comes in.
For every $i \in [d]$, we let $\tilde{\hull}_i$ be the convex hull for the $i$-th component $f_i$ of $\tilde{F}$
and $\hull_i$ be the convex hull constructed directly from the original inverted list $L_i$.
Next, we show that each $\tilde{\hull}_i$ can be efficiently obtained from $\hull_i$ during query time
so we only need to pre-compute the $\hull_i$'s.

We observe that when $L_i[b_i] \geq q_i \cdot \tilde{\tau}$, $f_i(L_i[b_i])$ equals a fixed value $q_i^2 \cdot \tilde{\tau}$
otherwise is proportional to $L_i[b_i]$.
As illustrated in Figure \ref{fig:projection} (left), the list of values $\{f_i(L_i[j])\}_{j \geq 0}$ is essentially obtained by
replacing the $L_i[j]$'s greater than $q_i \cdot \tilde{\tau}$ with $q_i \cdot \tilde{\tau}$.


The following can be shown using properties of convex hulls:
\vspace{-1mm}
\begin{lemma}\label{lem:projection}
For every $i \in [d]$, the convex hull $\tilde{\hull}_i$ is a subset of $\hull_i$
where an index $j_k$ of $\hull_i$ is in $\tilde{\hull}_i$ iff $k = 1$ or
\vspace{-1mm}
\begin{equation}\label{equ:projection}
\big(q_i \cdot \tilde{\tau} - L_i[j_k] \big) \ / \ j_k \geq \big(L_i[j_k] - L_i[j_{k+1}] \big) \ / \ (j_{k+1} - j_k) .
\vspace{-1mm}
\end{equation}
\end{lemma}
Lemma \ref{lem:projection} provides an efficient way to obtain each convex hull $\tilde{\hull}_i$
from the pre-computed $\hull_i$'s.
When a query $\mq$ is given, we perform a binary search on each $\hull_i$ to find the first $j_k \in \hull_i$
that satisfies (\ref{equ:projection}). Then $\tilde{\hull}_i$ is the set of indices $\{0, j_k, j_{k+1} \dots \}$.
We illustrate the construction in Figure \ref{fig:projection} (right).

Suppose that the maximum size of all $\hull_i$ is $h$.
The computation of the $\tilde{\hull}_i$'s adds an extra $\calO(d \log h)$ of overhead to the query processing time,
which is insignificant in practice since $h$ is likely to be much smaller than the size of the database. 

\vspace*{-2mm}
\section{Related work}\label{sec:related}
\vspace*{-2mm}


In this section, we present the main related work and defer the additional related work (e.g., dimensionality reduction, mass spectrometry search, inner product queries) to the appendix (Section~\ref{sec:addRelatedWork}).


\vspace*{-2mm}
\subsection{Cosine similarity search}\label{sec:relatedtechniques}

The cosine threshold querying studied in this work is a special case of the
\emph{cosine similarity search} (CSS) problem~\cite{Bayardo2007SUP,AnastasiuK14,Anastasiu2015PFP} mentioned in Section~\ref{sec:intro}.
We first survey the techniques developed for CSS.



\smallskip
\noindent
\textbf{LSH}.
A widely used technique for cosine similarity search is locality-sensitive hash (LSH)~\cite{Rajaraman2011,Andoni2015,Hu2017OPA,Indyk1998ANN,Tao2009QEH}.
The main idea of LSH is to partition the whole database into buckets using
a series of hash functions such that similar vectors have high probability to be in the same bucket.
However, LSH is designed for \emph{approximate} query processing, meaning that it is not guaranteed to return all the true results. In contrast, this work focuses on exact query processing which returns all the results.

\smallskip
\noindent \textbf{\textsf{TA}-family algorithms}.
Another technique for cosine similarity search is the family of $\TA$-like algorithms.
Those algorithms were originally designed for processing top-k ranking queries that find the top $k$ objects ranked according to an aggregation function
(see \cite{Ilyas2008STK} for a survey).
We have summarized the classic $\TA$ algorithm~\cite{Fagin2001OAA}, presented a baseline algorithm inspired by it,
and explained its shortcomings in Section~\ref{sec:intro}.
The Gathering-Verification framework introduced in Section \ref{sec:framework} captures
the typical structure of the $\TA$-family when applied to our setting.

The variants of \textsf{TA} (e.g., \cite{Guntzer2000OMQ,Bast2006IIO,Deshpande2008,Bruno02})
can have poor or no performance guarantee for cosine threshold queries
since they do not fully leverage the data skewness and the unit vector condition.
For example, G\"{u}ntzer et al. developed \textsf{Quick-Combine}~\cite{Guntzer2000OMQ}.
Instead of accessing all the lists in a lockstep strategy, it relies on a heuristic traversal strategy to
access the list with the highest rate of changes to the ranking function in a fixed number of steps ahead.
It was shown in \cite{fagin2003optimal} that the algorithm is not instance optimal.
Although the hull-based traversal strategy proposed in this paper roughly follows the same idea,
the number of steps to look ahead is variable and determined by the next convex hull vertex.
Thus, for decomposable functions, the hull-based strategy makes globally optimal decisions
and is near-optimal under the near-convexity assumption, while \textsf{Quick-Combine}
has no performance guarantee because of the fixed step size even when the data is near-convex.
Other \textsf{TA} variants are discussed in Appendix~\ref{sec:TAfamily}.

\smallskip
\noindent \textbf{COORD}.
Teflioudi et al. proposed the \textsf{COORD} algorithm based on inverted lists for CSS~\cite{Teflioudi2015LFR,Teflioudi2016}.
The main idea is to scan the whole lists but with an optimization to prune irrelevant entries using upper/lower bounds of the cosine similarity with the query.
Thus, instead of traversing the whole lists starting from the top, it scans only those entries within a feasible range.
We can also apply such a pruning strategy to the Gathering-Verification framework
by starting the gathering phase at the top of the feasible range.
However, there is no optimality guarantee of the algorithm. Also the optimization only works for high thresholds (e.g., 0.95), which are not always the requirement. For example, a common and well-accepted threshold in mass spectrometry search is 0.6, which is a medium-sized threshold, making the effect of the pruning negligible.

\smallskip
\noindent \textbf{Partial verification}.
Anastasiu and Karypis proposed a technique for fast verification of $\theta$-similarity between two vectors~\cite{AnastasiuK14}
without a full scan of the two vectors.
We can apply the same optimization to the verification phase of the Gathering-Verification framework.
Additionally, we prove that it has a novel near-constant performance guarantee in the presence of data skewness. See Appendix~\ref{sec:verification}.


\smallskip
\noindent \textbf{Other variants}.
There are several studies focusing on cosine similarity join to find out all pairs of vectors from the database such that their similarity exceeds a given threshold~\cite{Bayardo2007SUP,AnastasiuK14,Anastasiu2015PFP}.
However, this work is different since the focus is comparing to a given query vector \textbf{q} rather than join.
As a result, the techniques in \cite{Bayardo2007SUP,AnastasiuK14,Anastasiu2015PFP} are not directly applicable: (1) The inverted index is built online instead of offline, meaning that at least one full scan of the whole data is required, which is inefficient for search.
\yuliang{I removed the footnote. Please check whether it is okay.}
(2) The index in ~\cite{Bayardo2007SUP,AnastasiuK14,Anastasiu2015PFP} is built for a fixed query threshold, meaning that the index cannot be used for answering arbitrary query thresholds as concerned in this work.
The theoretical aspects of similarity join were discussed recently in \cite{Ahle2016CIP,Hu2017OPA}.

\vspace*{-2mm}
\subsection{Euclidean distance threshold queries}\label{sec:NNS}
\vspace*{-2mm}

The cosine threshold queries can also be answered by techniques for distance threshold queries (the threshold variant of nearest neighbor search)
in Euclidean space. This is because there is a one-to-one mapping between the cosine similarity $\theta$ and the Euclidean distance $r$ for unit vectors, i.e., $r = 2\sin(\arccos(\theta) / 2)$. Thus, finding vectors that are $\theta$-similar to a query vector is equivalent to finding the vectors whose Euclidean distance is within $r$.
Next, we review exact approaches for distance queries
while leaving the discussion of approximate approaches in the appendix (Section~\ref{sec:approSearch}).
There are four main types of techniques for exact approaches: tree-based indexing, pivot-based indexing, clustering, and dimensionality reduction (See Appendix \ref{app:dimensionality-reduction}).

\smallskip
\noindent\textbf{Tree-based indexing}.
Several tree-based indexing techniques (such as R-tree, KD-tree, Cover-tree~\cite{Beygelzimer2006}) were developed for range queries
(so they can also be applied to distance queries), see~\cite{Bohm2001SHS} for a survey.
However, they are not scalable to high dimensions (say thousands of dimensions as studied in this work)
due to the well known dimensionality curse issue~\cite{Weber1998QAP}.

\smallskip
\noindent\textbf{Pivot-based indexing}. The main idea is to
pre-compute the distances between data vectors and a set of selected pivot vectors. Then during query processing, use triangle inequalities to prune irrelevant vectors~\cite{Chen2017PMI,HristidisKP01}.
However, it does not scale in high-dimensional space as shown in \cite{Chen2017PMI} since it requires a large space to store the pre-computed distances.


\smallskip
\noindent \textbf{Clustering-based (or partitioning-based) methods}. The main idea of clustering is to partition the database vectors
into smaller clusters of vectors during indexing. Then during query processing,
irrelevant clusters are pruned via the triangle inequality~\cite{Samet2005FMM,Ramaswamy2011ACD}.
Clustering is an optimization orthogonal to the proposed techniques, as they can be used to process vectors
within each cluster to speed up the overall performance.

\vspace*{-2mm}
\section{Conclusion}\label{sec:conclusion}
\vspace*{-2mm}
In this work, we proposed optimizations to the index-based, \textsf{TA}-like algorithms for answering
the cosine threshold queries, which lie at the core of numerous applications.
The novel techniques include a complete and tight stopping condition computable incrementally in $\calO(\log d)$ time and
a family of convex hull-based traversal strategies with near-optimality guarantees
for a larger class of decomposable functions beyond cosine.
With these techniques, we show near-optimality first for inner-product threshold queries, 
then extend the result to the full cosine threshold queries using approximation.
These results are significant improvements over a baseline approach inspired by the classic $\TA$ algorithm.
In addition, we have verified with experiments on real data
the assumptions required by the near-optimality results. 

\vspace*{-2mm}
\bibliographystyle{plainurl}
\bibliography{paper.bib}

\appendix

\section*{Appendix}

\section{Review of $\TA$} \label{app:ta}

On a database $\db$ of $d$-dimensional vectors $\{\ms_1, \dots, \ms_n\}$, 
given a query monotonic scoring function $F : \mathbb{R}^d \mapsto \mathbb{R}$
and a query parameter $k$, the Threshold Algorithm ($\TA$) computes the $k$ database vectors with the highest score $F(\ms)$.
First, $\TA$ preprocesses the vector database by building an inverted index $\{L_i\}_{1 \leq i \leq d}$ 
where each $L_i$ is an inverted list contains pairs of $(\mathsf{ref}(\ms), \ms[i])$ where $\mathsf{ref}(\ms)$
is a reference of the vector $\ms$ and $\ms[i]$ is the $i$-th dimension $\ms$. Each $L_i$ is sorted in descending order of $\ms[i]$.
When a query $(F, k)$ arrives, $\TA$ proceeds as follows. It maintains a pointer $b$ starting from 0 to all the inverted lists
and increments $b$ iteratively. At each iteration:
\begin{itemize}
    \item Collect the set $\mathcal{C}$ of candidates of all references in $L_i$ upto position $b$ for all dimension $i$.
    Namely, $\mathcal{C} = \bigcup_{i=1}^d \{\mathsf{ref}(\ms) | (\mathsf{ref}(\ms), \ms[i]) \in L_i[0, \dots, b] \}$.
    \item Compute $F_k$ the $k$-th highest score $F(\ms)$ for all $\mathsf{ref}(\ms) \in \mathcal{C}$ by accessing $\ms$ in the database with the reference.
    \item If the score $F_k$ is no less than $F(L_1[b], \dots, L_d[b])$, return the top-$k$ highest score vectors in $\mathcal{C}$;
        otherwise continue to the next iteration with $b \leftarrow b + 1$.
\end{itemize}
By monotonicity of the function $F$, once the stopping condition is satisfied, it is guaranteed that no vector $\ms$ below the pointer $b$
can have $F(\ms)$ above the current $k$-th highest score. Thus the candidate set $\mathcal{C}$ contains the complete set of all 
the $k$ highest score vectors in the database.

\subsection{The $\TA$-inspired Baseline}

Next, we show the baseline stopping condition and traversal strategy inspired by the $\TA$ for cosine threshold queries.
As already reviewed in Section \ref{sec:intro}, the gathering phase stops when the cosine function score at the current position $\mb$
is below the input threshold $\theta$. Formally, the baseline stopping condition $\varphi_{\BL}$ is
\begin{equation}\label{equ:baseline}
\varphi_{\BL}(\mb) = \left( \sum_{i=1}^d q_i \cdot L_i[b_i] < \theta \right) .
\end{equation}
To determine the order of accessing the inverted lists, $\TA$ advances all the pointers at equal rate.
An obvious optimization is to move only the pointers whose dimension has non-zero values in $\mq$.
Let $\mathsf{nz} = \{ i | i \in [d] \land q_i > 0\}$ be the list of all non-zero dimensions and let $m = | \mathsf{nz} |$.
The baseline traversal strategy $\calT_{\BL}$ is the following:
\begin{equation}
\calT_{\BL}(\mb) = \mathsf{nz}[ \ |\mb| \ \mathsf{ mod } \ m \ ] ,
\end{equation}
so that each inverted list of a non-zero dimension is advanced at equal speed.
By the classic result of \cite{Fagin2001OAA},
\begin{theorem}
For inner product threshold queries (without the normalization constraint), the access cost of the $\TA$-inspired baseline is at most $m \cdot \opt$
where $m$ is the query's number of non-zero dimensions and $\opt$ is the optimal access cost.
\end{theorem}
For cosine threshold queries, since the baseline condition $\varphi_{\BL}$ is not tight (Theorem \ref{thm:TAnotTight}),
the same instance optimality bound does not hold.

\section{The Verification Phase}\label{sec:verification}


Next, we discuss optimizations in the verification phase where
each gathered candidate is tested for $\theta$-similarity with the query.
The naive approach of verification is to fully access all the non-zero entries of each candidate $\textbf{s}$ to compute
the exact similarity score $\csim(\mathbf{q}, \ms)$ and compare it against $\theta$, which takes $\calO(d)$ time.
Various techniques have been proposed \cite{Teflioudi2016,Anastasiu2015PFP,Li2017FFE} to
decide $\theta$-similarity by leveraging only partial information about the candidate vector
so the exact computation of $\csim(\mathbf{q}, \ms)$ can be avoided.
In this section, we revisit these existing techniques which we call \emph{partial verification}.
In addition, as a novel contribution, we show that in the presence of data skewness,
partial verification can have a near-constant performance guarantee (Theorem \ref{thm:verification}).

Informally, while a vector $\ms$ is scanned, based on what has been observed in $\ms$ so far,
it is possible to infer that
\begin{itemize}\parskip=0pt
\item[(1)] the similarity score $\csim(\textbf{q}, \ms)$ is certainly at least $\theta$ or
\item[(2)] the similarity score is certainly below $\theta$.
\end{itemize}
In either case, we can stop without scanning the rest of $\ms$ and
return an accurate verification result.
The problem is formally defined as follows:
\begin{problem}[Partial Verification]\label{pbm:verification}
A partially observed vector $\tilde{\ms}$ is a $d$-dimensional vector in $(\mathbb{R}_{\geq 0} \cup \{\bot\})^d$.
Given a query $\mathbf{q}$ and a partially observed vector $\tilde{\ms}$,
compute whether for every vector $\ms$ where $\ms[i] = \tilde{\ms}[i]$ for every $\tilde{\ms}[i] \neq \bot$,
it is $\csim(\ms, \mq) \geq \theta$.
\end{problem}
Intuitively, a partially observed vector $\tilde{\ms}$ contains the entries of a candidate $\ms$
either already observed during the gathering phase, or accessed for the actual values during the verification phase.
The unobserved dimensions are replaced with a null value $\bot$. We say that a vector $\ms$
is compatible with a partially observed one $\tilde{\ms}$ if $\ms[i] = \tilde{\ms}[i]$ for every dimension $i$ where $\tilde{\ms}[i] \neq \bot$.

The partial verification problem can be solved by computing an upper and a lower bound of the cosine similarity between $\ms$ and $\mathbf{q}$
when $\tilde{\ms}$ is observed. We denote by $ub(\tilde{\ms})$ and $lb(\tilde{\ms})$ the upper bound and lower bound, so
$ub(\tilde{\ms}) = \max_{\ms} \{\ms \cdot \mq\}$ and $lb(\tilde{\ms}) = \min_{\ms} \{ \ms \cdot \mq\}$
where the maximum/minimum are taken over all $\ms$ compatible with $\tilde{\ms}$.
By \cite{Teflioudi2016,Anastasiu2015PFP,Li2017FFE}, the upper/lower bounds can be computed as follows:
\begin{lemma}\label{lem:upper-lower}
Given a partially observed vector $\tilde{\ms}$ and a query vector $\mathbf{q}$,
\begin{align}
&ub(\tilde{\ms}) = \sum_{\tilde{\ms}[i] \neq \bot} \tilde{\ms}[i] \cdot \mathbf{q}[i] +
\sqrt{1 - \sum_{\tilde{\ms}[i] \neq \bot} \tilde{\ms}[i]^2} \cdot \sqrt{1 - \sum_{\tilde{\ms}[i] \neq \bot} \mathbf{q}[i]^2} \ , \\
&\text{and}  \nonumber \\
&lb(\tilde{\ms}) = \sum_{\tilde{\ms}[i] \neq \bot} \tilde{\ms}[i] \cdot \mathbf{q}[i] +
\sqrt{1 - \sum_{ \tilde{\ms}[i] \neq \bot} \tilde{\ms}[i]^2} \cdot \min_{\tilde{\ms}[i] = \bot} \mathbf{q}[i] \ .
\end{align}
\end{lemma}
\yuliang{I removed the proof of the Lemma since they have been provided in previous work.}

\begin{example}
Figure~\ref{fig:verifyExample} shows an example of computing the lower/upper bounds of a partially scanned vector $\ms$ with a query $\mq$.
Assume the first three dimensions have been scanned, then the lower bound and upper bound are computed as follows:
\begin{displaymath}
  \begin{aligned}
    lb &= 0.8 \times 0 + 0.4 \times 0.7 + 0.3 \times 0.5 = 0.43\\
    ub &= lb + \sqrt{1-0.8^2-0.4^2-0.3^2}\cdot\sqrt{1-0.7^2-0.5^2}=0.6
  \end{aligned}
\end{displaymath}
If the threshold $\theta$ is 0.7, then it is certain that $\ms$ is not $\theta$-similar to $\mq$ because $0.6 < 0.7$.
The verification algorithm can stop and avoid the rest of the scan.
Note that when $\mq$ is sparse, most of the time we have $lb(\tilde{\ms}) = \sum_{\tilde{\ms}[i] \neq \bot} \tilde{\ms}[i] \cdot \mathbf{q}[i]$
due to the existence of 0's.
\end{example}

\begin{figure}[htbp]
\centering
\renewcommand{\tabcolsep}{0.1mm}
\includegraphics[width=0.6\columnwidth]{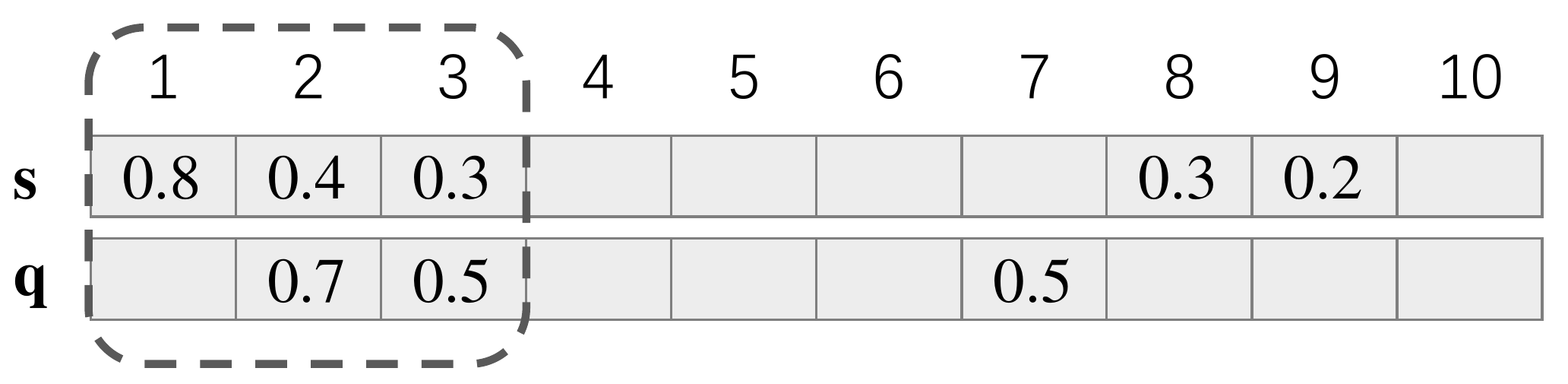}
\caption{An example of the lower and upper bounds}\label{fig:verifyExample}
\end{figure}

%

\smallskip
\noindent \textbf{Performance Guarantee. }
Next, we show that when the data is skewed,  partial verification
achieves a strong performance guarantee.
We assume that each candidate $\ms$ is stored in the database such that
$\ms[1] \geq \ms[2] \geq \dots \geq \ms[d]$ and we scan $\ms$ sequentially starting from $\ms[1]$.
We notice that partial verification saves data accesses when there is a gap between the true similarity score
and the threshold $\theta$. Intuitively, when this gap is close to 0, both the upper and lower bounds
converge to $\theta$ as we scan $\ms$ so we might not stop until the last element.
When the gap is large and $\ms$ is skewed, meaning that the first few values account for most of the candidate's weight,
then the first few $\ms[i] \cdot \mathbf{q}[i]$ terms can provide large enough information
for the lower/upper bounds to decide $\theta$-similarity.
Formally,
\begin{theorem} \label{thm:verification}
Suppose that a vector $\ms$ is skewed: there exists an integer $k \le d$ and constant value $c$ such that
$\sum_{i=1}^k \ms[i]^2 \geq c$.
For every query $(\mathbf{q}, \theta)$,
if $|\csim(\ms, \mathbf{q}) - \theta| \geq \sqrt{1 - c}$,
then the number of accesses for verifying $\ms$ is at most $k$.
\end{theorem}

Equivalently, Theorem \ref{thm:verification} says that if the true similarity is at least $\delta$ off the threshold $\theta$
(i.e. $|\csim(\ms, \mq) - \theta| \geq \delta$ for $\delta > 0$), then it is only necessary to access $k$ entries of $\ms$
with the smallest $k$ that satisfies $\sum_{i=1}^{k} \ms[i]^2 \ge 1 - \delta^2$.
For example, if $\delta = 0.1$ and the first 20 entries of a candidate $\ms$ account for >99\% of $\sum_{i=1}^d \ms[i]^2$,
then it takes at most 20 accesses for verifying $\ms$.


\begin{proof}
\textbf{Case one:} We first consider the case where
$\csim(\ms, \mathbf{q}) - \theta \geq \sqrt{1 - c}$.
In this case, we need to show  $\sum_{i=1}^{k} \ms[i] \cdot \mathbf{q}[i]\ge\theta$.
Since $\csim(\ms, \mathbf{q}) - \theta \geq \sqrt{1 - c}$,
which means $$\sum_{i=1}^k \ms[i] \cdot \mathbf{q}[i] + \sum_{i=k+1}^d \ms[i] \cdot \mathbf{q}[i] \ge \theta + \sqrt{1 - c} \ , $$
it suffices to show that $\sum_{i=k+1}^d \ms[i] \cdot \mathbf{q}[i] \le \sqrt{1 - c}$.
This can be obtained by
\begin{align*}
  \sum_{i=k+1}^d \ms[i] \cdot \mathbf{q}[i] &\le \sqrt{\sum_{i=k+1}^d \ms[i]^2}\cdot \sqrt{\sum_{i=k+1}^d \mathbf{q}[i]^2}  \\
  &=\sqrt{1-\sum_{i=1}^k \ms[i]^2}\cdot\sqrt{1-\sum_{i=1}^k \mq[i]^2}  \\
  &\le \sqrt{1-c} \ .
\end{align*}




%
\smallskip
\noindent
\textbf{Case two:}
We then consider the case where
$ \csim(\ms, \mathbf{q}) - \theta \leq -\sqrt{1 - c}$.
In this case, we need to show 
$$
  \sum_{i=1}^k \ms[i] \cdot \mathbf{q}[i] + \sqrt{1 - \sum_{i=1}^k \ms[i]^2}  \cdot \sqrt{1 - \sum_{i=1}^k \mathbf{q}[i]^2} \leq \theta \ .
$$
The first term of the LHS is bounded by $\sum_{i=1}^d \ms[i] \cdot \mathbf{q}[i] \leq \theta - \sqrt{1 - c}$.
The second term of the LHS is bounded by $\sqrt{1 - \sum_{i=1}^k \ms[i]^2} \leq \sqrt{1 - c}$.
Adding together, the LHS is bounded by $\theta$.
\end{proof}



\begin{example}
Figure~\ref{fig:verification}a shows the number of accesses for each candidate during the verification phase of
a real query with a vector having 100 non-zero values
and $\theta = 0.6$. The result shows that for most candidates, the number of accesses is much smaller than 100.
In particular, 55.9\% candidates need less than five accesses and 93.1\% candidates need less than 30 accesses.
This is because as shown in Figure \ref{fig:verification}b,
only 0.23\% of candidates have true similarity within $\pm$0.2 compared to $\theta$ (the range [0.4, 0.8]).
The rest of the candidates, according to Theorem \ref{thm:verification}, can be verified in a small number of steps.
%

%
%
\end{example}

\begin{figure}[tbp]
\centering
\renewcommand{\tabcolsep}{0.1mm}
\begin{tabular}{cccc}
\includegraphics[width=0.4\columnwidth]{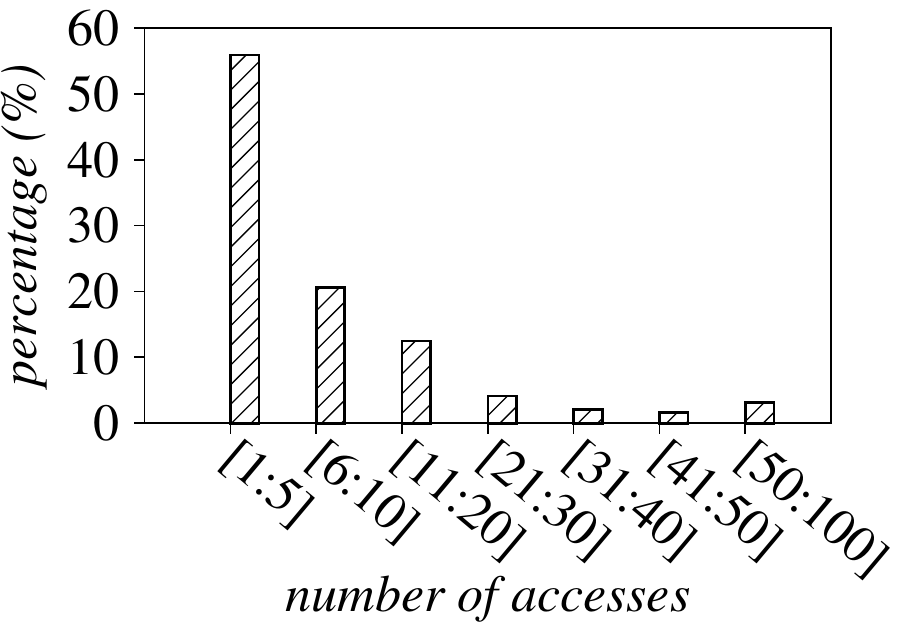}&
\includegraphics[width=0.4\columnwidth]{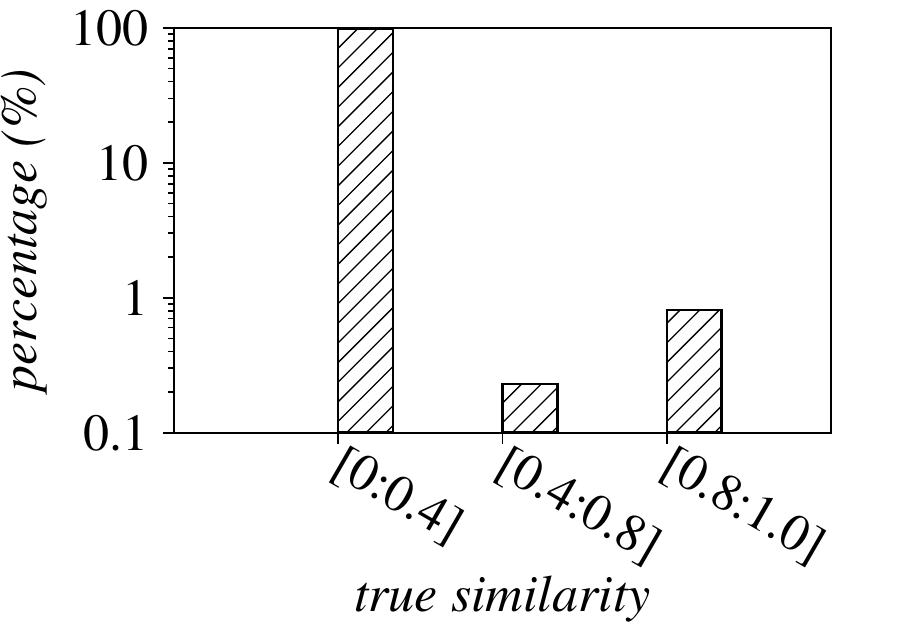}\\
(a) number of accesses & (b)  true similarities
\end{tabular}
\caption{{The distribution of number of accesses and true similarities}}\label{fig:verification}
\end{figure}

\section{The baseline stopping condition is not tight}\label{sec:TANotTight}

%

\begin{theorem}\label{thm:TAnotTight}
The \emph{baseline} stopping condition
\begin{equation}
 \varphi_{\BL}(\mb) = \left( \sum_{i=1}^{d} q_i \cdot L_i[b_i] < \theta \right)
\end{equation}
is complete but not tight.
\end{theorem}


\begin{proof}
For every position vector $\mb$, $\varphi_{\BL}(\mb) = \mathtt{True}$ implies
$\mq \cdot L[\mb] < \theta$. So for every $\ms < L[\mb]$, we also have $\mq \cdot \ms < \theta$ so
$\varphi_{\BL}$ is complete. 

To show the non-tightness, it is sufficient to show that for some position vector $\mb$ where $\mb$ is complete,
$\varphi_{\BL}(\mb)$ is $\mathtt{False}$ so the traversal continues.

We illustrate a counterexample in Figure \ref{fig:exampleTight} with two dimensions (i.e., $d = 2$).
Given a query $\textbf{q}$, all possible $\theta$-similar vectors form a hyper-surface defining the set
$\mathsf{ans} = \{\ms | \ \|\ms\| = 1, \sum_{i=1}^{d} q_i \cdot s_i \ge \theta\}$.
In Figure~\ref{fig:exampleTight}, $\|\ms\| = 1$ is the circular surface and $\sum_{i=1}^{d} q_i \cdot s_i \ge \theta$
is a half-plane so the set of points $\mathsf{ans}$ is the arc $\widehat{AB}$.

By definition, a position vector $\mb$ is complete if the set $\{\ms | \ms < L[\mb] \}$ contains no point in $\mathsf{ans}$.
A position vector $\mb$ satisfies $\varphi_{\BL}$ iff the point $L[\mb]$ is above the plane $\sum_{i=1}^{d} q_i \cdot s_i = \theta$.
It is clear from Figure \ref{fig:exampleTight} that if the point $L[\mb]$ locates at the region BCD, then
$\{\ms | \ms < L[\mb] \}$ contains no point in $\widehat{AB}$ and is above the half-plane $\sum_{i=1}^{d} q_i \cdot s_i = \theta$.
There exists a database of 2-d vectors such that $L[\mb]$ resides in the BCD region for some position $\mb$, 
so the stopping condition $\varphi_{\BL}$ is not tight.
\end{proof}

\vspace*{-4mm}
\smallskip
\noindent \textbf{Remark.} The baseline stopping condition $\varphi_{\BL}$ is not tight because it does not 
take into account that all vectors in the database are unit vectors. 
In fact, one can show that $\varphi_{\BL}$ is tight and complete for inner product queries where the unit vector assumption is lifted.
In addition, since $\varphi_{\BL}$ is not tight, any traversal strategy that works with $\varphi_{\BL}$
has no optimality guarantee in general since there can be a gap of arbitrary size 
between the stopping position by $\varphi_{\BL}$ and the one that is tight (i.e. there can be arbitrarily many points in the region $BCD$).

\begin{figure}[tbp]
  \centering
  \includegraphics[width=0.5\textwidth]{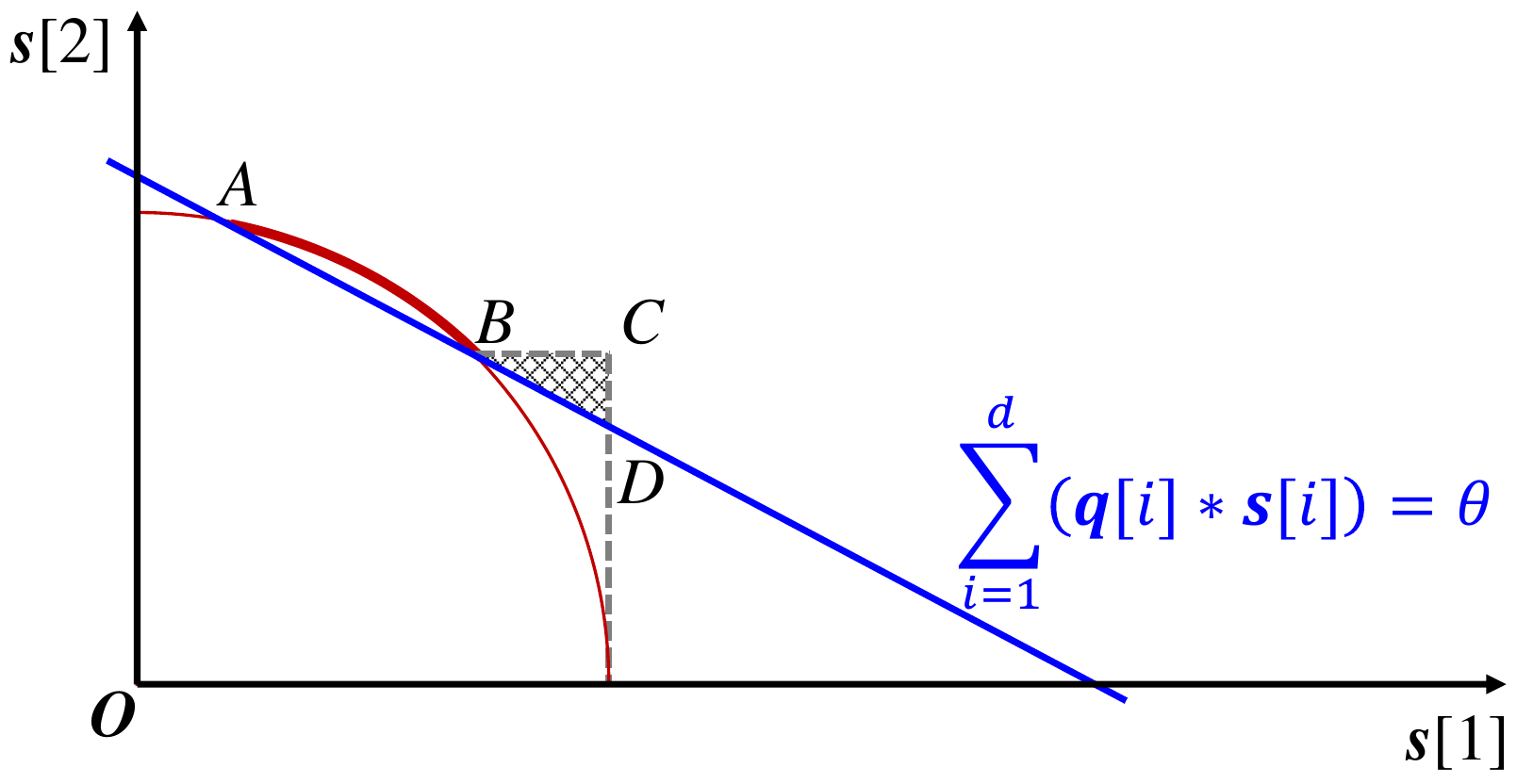}
  \caption{A 2-d example of $\varphi_{\BL}$'s non-tightness}\label{fig:exampleTight}
\vspace*{-2mm}
\end{figure}

\section{Efficient computation of $\varphi_{\TC}$ with incremental maintenance}\label{sec:incremental}

Testing the tight and complete condition $\varphi_{\TC}$ requires solving $\tau$ in Theorem \ref{thm:qiLibi}, 
for which a direct application of the bisection method takes $\calO(d)$ time.
Next, we provide a more efficient algorithm based on incremental maintenance which takes only
$\calO(\log d)$ time for each test of $\varphi_{\TC}$.


According to the proof of Theorem~\ref{thm:qiLibi},
\begin{equation}
    s_i=
\begin{cases}
    L_i[b_i],  & \tau \geq \dfrac{L_i[b_i]}{q_i} ; \\
    q_i \cdot \tau, & \text{otherwise}.
\end{cases}
\end{equation}
Wlog, suppose $\frac{L_1[b_1]}{q_1} \le \dots \le \frac{L_d[b_d]}{q_d}$ and
$\tau$ is in the range $[\frac{L_k[b_k]}{q_k},  \frac{L_{k+1}[b_{k+1}]}{q_{k+1}}]$ for some $k$. 
We have $s_i = L_i[b_i]$ for every $1 \leq i \leq k$ and $s_i = q_i \cdot \tau$ for $k < i \leq d$. 
So if we let $\mathsf{eval}(k, \tau)$ be the function
\begin{equation}
\mathsf{eval}(k, \tau) = \sum_{i=1}^d s_i^2 = \sum_{i=1}^k L_i[b_i]^2 + \sum_{i=k+1}^d q_i^2 \cdot \tau^2 ,
\end{equation}
then for the largest $k$ such that $\mathsf{eval}(k, L_k[b_k] / q_k) \leq 1$, $\tau$ can be computed by solving
\begin{align}\label{equ:tau2}
 \sum_{i=1}^k L_i[b_i]^2 + \sum_{i=k+1}^d q_i^2 \cdot \tau^2 = 1 
\Rightarrow \quad  \tau = \left(\dfrac{1 - \sum_{i=1}^k L_i[b_i]^2}{1-\sum_{i=1}^k q_i^2} \right)^{1/2} .
\end{align}
Then, $\score(L[\mb])$ can be computed as follows:
\begin{equation}\label{equ:dms}
  \score(L[\mb]) = \sum_{i=1}^k L_i[b_i] \cdot q_i + (1-\sum_{i=1}^k q_i^2) \cdot \tau \ .
\end{equation}
Computing $\score(L[\mb])$ using the above approach directly requires that the $L_i[b_i]$'s 
are sorted in each step by $L_i[b_i]/q_i$, which requires $\calO(d \log d)$ time. 
However, this is still too expensive as the stopping condition is checked in every step.
Fortunately, we show that $\score(L[\mb])$ can be incrementally maintained in $\calO(\log d)$ time as we describe below.

We use a binary search tree (BST) to maintain an order of the $L_i$'s sorted by $L_i[b_i] / q_i$.
The BST supports the following two operations: 
\begin{itemize}\parskip=0in\itemsep=0in
\item $\mathsf{update(i)}$: update $L_i[b_i] \rightarrow L_i[b_i + 1]$
\item $\mathsf{compute()}$: return the value of $\score(L[\mb])$
\end{itemize}
The $\mathsf{compute()}$ operation essentially performs the binary search of finding the largest $k$ mentioned above.
To ensure $\calO(\log d)$ running time, 
we observe that from Equation (\ref{equ:tau2}) and (\ref{equ:dms}), for any $k$, $\score(L[\mb])$ can be computed if $\sum_{i=1}^{k}L_i[b_i]\cdot q_i$, $\sum_{i=1}^{k}L_i[b_i]^2$, and $\sum_{i=1}^{k}q_i^2$ are available. 
Let $\TT$ be the BST and each node in $\TT$ is denoted as an integer $\mathsf{n}$,
meaning that the node represents the list $L_{\mathsf{n}}$.
We denote by $\mathsf{subtree(n)}$ the subtree of $\TT$ rooted at node $\mathsf{n}$, and maintain the following information for each node $\mathsf{n}$:
\begin{itemize}\parskip=0in\itemsep=0in
\item $\mathsf{n.key}$: $L_n[b_n] / q_n$,
\item $\mathsf{n.LQ}$: $\sum_{i \in \mathsf{subtree(n)}} L_i[b_i] \cdot q_i$,
\item $\mathsf{n.Q2}$: $\sum_{i \in \mathsf{subtree(n)}} q_i^2$ and
\item $\mathsf{n.L2}$: $\sum_{i \in \mathsf{subtree(n)}} L_i[b_i]^2$ .
\end{itemize}

\begin{figure}[tbp]
\centering
\includegraphics[width=.6\columnwidth]{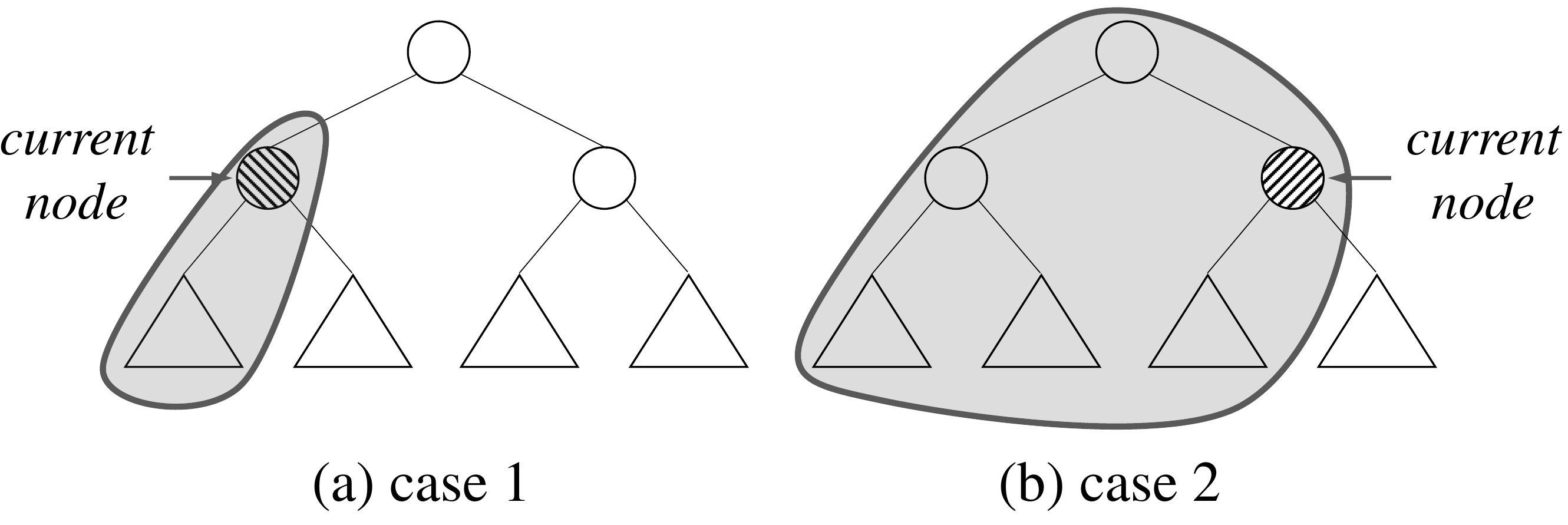}
\caption{{An example of incremental maintenance}}\label{fig:incremental}
\end{figure}

Thus, whenever there is a move on the list $L_i$, the key of the node $\mathsf{i}$ (i.e., $L_i[b_i]/q_i$) will be updated. Then we can remove the node  $\mathsf{i}$ from the tree and insert it again using the new key, which takes $\calO(\log d)$ time (with all the associated values being updated as well). To compute $\score(L[\mb])$, 
we need to handle two cases shown in Figure~\ref{fig:incremental}. In the first case, all the required information for computing  $\score(L[\mb])$ are stored in the current node and its left subtree, while in the second case, 
the required information needs to be passed from the parent node to the current node to compute $\score(L[\mb])$.

Algorithm~\ref{alg:compute} shows the details to incrementally computes $\score(L[\mb])$ mentioned in Section~\ref{sec:incremental}.
For each node $\mathsf{n}$, we denote by $\mathsf{n.left}$ ($\mathsf{n.right}$) the left (right) child of $\mathsf{n}$.

\begin{algorithm}[h!]
$\mathsf{(LQ\_parent, Q2\_parent, L2\_parent)} \leftarrow (0, 0, 0)$\;
$\score(L[\mb]) \leftarrow 1$ \tcp*{$\score(L[\mb]) = 1$ if $\forall i \ \tau \leq L_i[b_i]/q_i$}
$\mathsf{n} \leftarrow \mathsf{root}(T)$\;

\While {$\mathsf{n} \neq \mathsf{null}$} {
$\mathsf{LQ} \leftarrow \mathsf{LQ\_parent + n.left.LQ} + L_{\mathsf{n}}[b_{\mathsf{n}}] \cdot q_{\mathsf{n}}$\;
$\QQ \leftarrow \mathsf{Q2\_parent + n.left.Q2} + q_{\mathsf{n}}^2$\;
$\LL \leftarrow \mathsf{L2\_parent + n.left.L2} + L_{\mathsf{n}}[b_{\mathsf{n}}]^2$\;
$\mathsf{f(n)} \leftarrow \mathsf{LQ} + (1 - \QQ) \cdot \mathsf{n.key}$\;

\If {$\mathsf{f(n)} \leq 1$} {
  $\tau \leftarrow \left((1 - \LL) / (1 - \QQ)\right)^{1/2} $\;
  $\score(L[\mb]) \leftarrow \mathsf{LQ} + (1 - \QQ) \cdot \tau $\;
  $\mathsf{n} \leftarrow \mathsf{n.left}$\;
}\Else {
$\mathsf{LQ\_parent} \leftarrow \mathsf{LQ\_parent + n.left.LQ} + L_{\mathsf{n}}[b_{\mathsf{n}}] \cdot q_{\mathsf{n}}$\;
$\mathsf{Q2\_parent} \leftarrow \mathsf{Q2\_parent + n.left.Q2} + q_{\mathsf{n}}^2$\;
$\mathsf{L2\_parent} \leftarrow \mathsf{L2\_parent + n.left.L2} + L_{\mathsf{n}}[b_{\mathsf{n}}]^2$\;
  $\mathsf{n} \leftarrow \mathsf{n.right}$\;
}
}
\Return $\score(L[\mb])$\;
 \caption{\textsf{compute(T)}}\label{alg:compute}
\end{algorithm}


\section{Traversal starting from the middle} \label{app:middle}
Next, we consider a variant of the Gathering-Verification algorithm starting the traversal 
from the middle of the inverted lists instead of starting from the top like the classic $\TA$.
In the most general setting, the starting position can be anywhere in the inverted lists, 
and two pointers (upper and lower) are used to indicate the traversed range for each inverted list. 
At each iteration, in addition to choosing the list to be traversed, the traversal strategy also decides
whether the upper or the lower pointer of the selected list needs to be moved.
We can show that in this setting, deciding a tight and complete stopping condition is {\sc np-hard} thus intractable.

We start with some definitions.
A \emph{configuration} of the traversal is a pair of position vector $(\ma, \mb)$
where $\ma$ and $\mb$ are the upper and the lower position vector of the inverted lists $\{L_i\}_{i \in [d]}$ respectively.
At each configuration, it is guaranteed that all vectors $\ms$ satisfying $L_i[b_i] \leq s_i \leq L_i[a_i]$ for some dimension $i$
have been collected in the candidate set. 
So the tight stopping condition needs to test whether all the remaining vector $\ms$ can be $\theta$-similar to the query $\mq$.
For each dimension $i$, the condition that $s_i$ needs to satisfy is 
$$ s_i \geq L_i[a_i] \lor s_i \leq L_i[b_i] \ , $$
which can be written into an equivalent quadratic form
$$ (s_i - L_i[a_i]) \cdot (s_i - L_i[b_i]) \geq 0 $$
since $L_i[a_i] \geq L_i[b_i]$.

Thus, the tight stopping condition is equivalent to the following quadratic program:
\begin{equation} \label{equ:maximization2}
\begin{array}{ll@{}ll@{}ll}
\text{maximize}   & \displaystyle\sum_{i=1}^{d} s_i \cdot q_{i}    & \\
\text{subject to} & \displaystyle\sum_{i=1}^d s_i^2 = 1,           & \\
                  & (s_i - L_i[a_i]) \cdot (s_i - L_i[b_i]) \geq 0, & \quad \text{ for } i \in [d].
\end{array}
\end{equation}

Let 
$$ F(\ms, \mathbf{\mu}, \lambda) = \sum_{i=1}^d s_i q_i + \sum_{i=1}^d \mu_i(s_i - L_i[a_i]) \cdot (s_i - L_i[b_i]) - \lambda \left(\sum_{i=1}^d s_i^2 - 1\right) $$
be the Lagrangian of (\ref{equ:maximization2}) where 
$\lambda \in \mathbb{R}$ and $\mathbf{\mu} \in \mathbb{R}^d$ are the Lagrange multipliers.
Then, the KKT conditions of (\ref{equ:maximization2}) are
\begin{equation*}
\begin{array}{ll@{}ll}
\nabla_{\ms} F(\ms, \mu, \lambda) = 0 & \text{(Stationarity)} \\
\mu_i \geq 0 , \ \forall \ i \in [d] & \text{(Dual feasibility)} \\
\mu_i (s_i - L_i[a_i]) (s_i - L_i[b_i]) = 0, \ \forall \ i \in [d] & \text{(Complementary slackness)}
\end{array}
\end{equation*}
in addition to the Primal feasibility in (\ref{equ:maximization2}).

By the Complementary slackness condition, either $\mu_i = 0$ or $(s_i - L_i[a_i])(s_i - L_i[b_i]) = 0$.
If $\mu_i = 0$, from the Stationarity condition, we have
$$ q_i + \mu_i (s_i / 2 - L_i[a_i] - L_i[b_i]) - \lambda s_i = 0  $$
so $s_i = q_i / \lambda$.

If $\mu_i \neq 0$, then $(s_i - L_i[a_i])(s_i - L_i[b_i]) = 0$ so $s_i = L_i[a_i]$ or $s_i = L_i[b_i]$.
Since $\ms \cdot \mq$ is maximized when $s_i$ is proportional to $q_i$ and $s_i \geq L_i[a_i] \lor s_i \leq L_i[b_i]$,
the value of $s_i$ is determined as follows:
\begin{itemize} \parskip=0pt
\item if $q_i / \lambda \in [L_i[b_i], L_i[a_i]]$, then $s_i$ equals either the upper bound $L_i[a_i]$ or the lower bound $L_i[b_i]$;
\item otherwise, $s_i = q_i / \lambda$.
\end{itemize}
Here, we can see that solving the quadratic program (\ref{equ:maximization2})
is more difficult than the quadratic program (\ref{equ:maximization}) in Section \ref{sec:stop}
where the traversal starts from the top of the lists.
Intuitively, when $s_i \neq q_i / \lambda$, unlike the solution for (\ref{equ:maximization}), 
there are still two possible options: the upper bound $L_i[a_i]$ and the lower bound $L_i[b_i]$, 
which lead to exponentially many combinations. 
One can easily show that even when $\lambda$ is fixed, 
computing the optimal $\ms$ is {\sc np-hard} by reduction from the subset sum problem \cite{karp1972reducibility}.
This is because when $\lambda$ is fixed, the $s_i$'s that equal to $q_i / \lambda$ are fixed.
When $L_i[b_i] = 0$ for all $i$, checking whether $\sum_{i=1}^d s_i^2 = 1$ is essentially checking whether
there exists a subset of $\{L_i[a_i]^2\}_{i \in [d]}$ whose sum is some fixed constant depending on $\lambda$.

\section{Proof of Theorem \ref{thm:optimal}}\label{thm:proof-optimal}
\begin{proof}
Let $\{\mb_t\}_{1 \leq t \leq k}$ be the sequence of position vectors produced by the strategy $\calT_{\mathsf{MR}}$.

Since each $\Delta_i$ is non-increasing and the strategy $\calT_{\mathsf{MR}}$ chooses the dimension $i$
with the maximal $\Delta_i[b_i]$, then at each step $t$, the multiset
$\{ \Delta_i[j] | 1 \leq i \leq d, 0 \leq j \leq \mb_t[i] \}$
contains the first $t$ largest values of all the $\Delta_i[j]$'s from the multiset
$\{ \Delta_i[j] | 1 \leq i \leq d, 0 \leq j < |L_i| \}$.
Since the score $F(L[\mb_t])$ equals
$$\sum_{i=1}^d f_i(L_i[0]) - \sum_{i=1}^d \sum_{j=1}^{\mb_t[i]} \Delta_i[j] \ , $$
it follows that for each $\mb_t$ of $\calT_{\mathsf{MR}}$,
the score $F(L[\mb_t])$ is the lowest score possible for any position vector reachable in $t$ steps.
Thus, if the optimal access cost $\opt$ is $t$ with an optimal stopping position $\mb_{\opt}$,
then $\mb_t$, the $t$-th position of $\calT_{\MR}$, satisfies that
$F(L[\mb_t]) \leq F(L[\mb_{\opt}]) < \theta$. So $\calT_{\mathsf{MR}}$ is optimal.
\end{proof}

\section{Proof of Lemma \ref{lem:bstar}} 
\label{sec:proof-near-optimal}
\begin{proof}
We construct a new collection of inverted lists $\{\tilde{L}_i\}_{1 \leq i \leq d}$ from the original lists as follows.
For every $i$ and every pair of consecutive indices $j_k, j_{k+1}$ of $\hull_i$, we assign
$$ \tilde{L}_i[j] = f_i(L_i[j_k]) + (j - j_k) \cdot \tilde{\Delta}_i[j] \quad \text{ for every } j \in [j_k, j_{k+1}].$$
Intuitively, we construct each $\tilde{L}_i$ by projecting the set of 2D points $\{(i, f_i(L_i[j]))\}_{j \in [j_k, j_{k+1}]}$
onto the line passing through the two boundary points with index $j_k$ and $j_{k+1}$, which is
essentially projecting the set of points onto the piecewise linear function defined by the convex hull vertices in $\hull_i$
(See Figure \ref{fig:convex}(c) for an illustration).
The new $\{\tilde{L}_i\}_{1 \leq i \leq d}$ satisfies the following properties.
\begin{itemize}\parskip=0pt
\item[(i)] By the construction of each convex hull $\hull_i$, we have $\tilde{L}_i[j] \leq f_i(L_i[j])$ for every $i$ and $j$.
\item[(ii)] For every boundary position $\mb$, we have $\tilde{L}[\mb] = F(L[\mb])$ since for every index $j$ on a convex hull $\hull_i$,
$\tilde{L}[j] = f_i(L_i[j])$.
\item[(iii)] The collection $\{\tilde{L}_i\}_{1 \leq i \leq d}$ is ideally convex\footnote{where each $f_i$ of the decomposable function $F$ is the identity function}.
In addition, the strategy $\calT_{\HL}$ produces exactly the same sequence reduced by the max-reduction strategy $\calT_{\MR}$
when $\{\tilde{L}_i\}_{1 \leq i \leq d}$ is given as the input. By the same analysis for Theorem \ref{thm:optimal},
for every position vector $\mb$ generated by $\calT_{\HL}$, $\tilde{L}[\mb]$ is minimal among all position vectors reached within $|\mb|$ steps.
\end{itemize}
Combining (ii) and (iii), for every boundary position vector $\mb$ generated by $\calT_{\HL}$ and every $\mb^*$ where $|\mb^*| = \mb$,
we have $ F(L[\mb]) = \tilde{L}[\mb] \leq \tilde{L}[\mb^*]$.
Finally, by (i) and since $F$ is non-decreasing, we have $\tilde{L}[\mb^*] \leq F(L[\mb^*])$ so $F(L[\mb]) \leq F(L[\mb^*])$ for every $\mb^*$.
\end{proof}

\section{Estimation of $\epsilon_1 + \epsilon_2$} 
\label{app:estimation}

\smallskip
\noindent
\textbf{For $\epsilon_1$}: A trivial upper bound of $\epsilon_1$ is $\tilde{\tau} - 1$ since the initial value of $\tau$ is 1
and the gap $\tilde{F}(L[\mb]) - \score(L[\mb])$ is maximized when $\mb = \mathbf{0}$. 
This upper bound can be improved as follows. We notice that in the proof of Theorem \ref{thm:near-optimal-approx},
given $|\mb^*| = |\mb_l|$, we need to have $F(L[\mb^*]) \geq \theta - \epsilon_1 - \epsilon_2$ for every such $\mb^*$,
which is equivalent to requiring that this property holds for the $\mb^*$ that minimizes $F(L[\mb^*])$ given $|\mb^*| = |\mb_l|$.
This $\mb^*$ satisfies that $F(L[\mb^*]) \leq F(L[\mb_l])$ and $F(L[\mb_l])$ is known when the query is executed.
This upper bound of $F(L[\mb^*])$ implies a lower bound of $\tau$ at position $\mb^*$, which also implies
the following lower bound of $\score(L[\mb^*]) - \tilde{F}(L[\mb^*])$:

\begin{lemma}\label{lem:esp}
Let $\mb$ be an arbitrary position vector and let $\mb^*$ be the position vector such that
$$ \mb^* = \argmin_{\mb': |\mb| = |\mb'|} \{\score(L[\mb'])\} .$$
Then
$$ (\dag) \quad \score(L[\mb^*]) - \tilde{F}(L[\mb^*]) \geq \min\{0, 1 / \score(L[\mb]) - \tilde{\tau} \} . $$
where $\tilde{F}$ is the decomposable function where each component is
$f_i(x) = \min\{\tilde{\tau} \cdot q_i, x \} \cdot q_i$ for constant $\tilde{\tau}$ and every $1 \leq i \leq d$.
\end{lemma}

\begin{proof}
Let $\tau^*$ be the value of $\tau$ at $L[\mb^*]$. We consider two cases separately: $\tau^* \geq \tilde{\tau}$
and $\tau^* < \tilde{\tau} $.

\smallskip
\textbf{Case One: } When $\tau^* \geq \tilde{\tau}$, since each $f_i$ increases as $\tilde{\tau}$ increases,
we have $$f_i(x) = \min\{\tilde{\tau} \cdot q_i, x \} \cdot q_i \leq \min\{\tau^* \cdot q_i, x \} \cdot q_i$$ 
thus $\score(L[\mb^*]) \geq \tilde{F}(L[\mb^*])$.

\smallskip
\textbf{Case Two: } Suppose $\tau^* < \tilde{\tau}$. We show the following Lemma.
\begin{lemma}\label{lem:theta}
For every position $\mb$ and constant $c$, $\score(L[\mb]) < c$ implies that the $\tau$ at $L[\mb]$ is at least $1 / c$.
\end{lemma}

Recall the notations $\mathsf{LQ} = \sum_{i: L_i[b_i] < \tau \cdot q_i} L_i[b_i] \cdot q_i$,
$\QQ = \sum_{i: L_i[b_i] < \tau \cdot q_i} q_i^2$ and \\ $\LL = \sum_{i: L_i[b_i] < \tau \cdot q_i} L_i[b_i]^2$.
Then $\tau$ satisfies that
\begin{equation}\label{equ:LQ}
\LQ + (1 - \QQ) \cdot \tau < c
\end{equation} and
\begin{equation}\label{equ:LL}
\LL + (1 - \QQ) \cdot \tau^2 = 1 .
\end{equation}
By rewriting Equation (\ref{equ:LL}), we have $(1 - \QQ) = (1 - \LL) / \tau^2$. Plug this into (\ref{equ:LQ}), we have
$\LQ + (1 - \LL) / \tau < c$. Since $L_i[b_i] < \tau \cdot q_i$, we have $\LQ > \LL / \tau$ so
$1 / \tau < c$ which means $\tau > 1 / c$. This completes the proof of Lemma \ref{lem:theta}.

We know that since $\mb^*$ minimizes $\score(L[\mb'])$ among all $\mb'$ with $|\mb'| = |\mb|$,
we have $\score(L[\mb^*]) \leq \score(L[\mb])$. By Lemma \ref{lem:theta}, we have $\tau^* \geq 1 / \score(L[\mb])$.

Now consider $\score(L[\mb^*]) - \tilde{F}(L[\mb^*])$. Since $\tau^* < \tilde{\tau}$, $\score(L[\mb^*]) - \tilde{F}(L[\mb^*])$
can be written as the sum of the following 3 terms:
\begin{itemize} \parskip=0pt
\item $\sum_{i : L_i[b_i] / q_i < \tau^*} (L_i[b_i] \cdot q_i - L_i[b_i] \cdot q_i)$, which is always 0,
\item $\sum_{i : \tau^* \leq L_i[b_i] / q_i < \tilde{\tau}} (q_i^2 \cdot \tau^* - L_i[b_i] \cdot q_i)$ and
\item $\sum_{i : \tilde{\tau} \leq L_i[b_i] / q_i} (q_i^2 \cdot \tau^* - q_i^2 \cdot \tilde{\tau})$.
\end{itemize}
In the second term, since each $L_i[b_i]$ is at most $q_i \cdot \tilde{\tau}$, so this term is greater than or equal to
$\sum_{i : \tau^* \leq L_i[b_i] / q_i < \tilde{\tau}} (q_i^2 \cdot \tau^* - q_i^2 \cdot \tilde{\tau})$.
Adding them together, we have $\score(L[\mb^*]) - \tilde{F}(L[\mb^*])$ is lower bounded by
$$ \sum_{i : \tau^* \leq L_i[b_i] / q_i} q_i^2 \cdot (\tau^* - \tilde{\tau}) $$
which is greater than or equal to $\tau^* - \tilde{\tau}$. Since $\tau^* \geq 1 / \score(L[\mb])$, we have
$\score(L[\mb^*]) - \tilde{F}(L[\mb^*]) \geq 1 / \score(L[\mb]) - \tilde{\tau}$.
Finally, combining case one and two together completes the proof of Lemma \ref{lem:esp}.
\end{proof}
Thus, when the query is given, $\epsilon_1$ is at most $\max\{0, \tilde{\tau} - 1 / \score(L[\mb_l])\}$.

\smallskip
\noindent
\textbf{For $\epsilon_2$}: In general, there is no upper bound for $\tau$ since it can be as large as $L_i[b_i] / q_i$ for some $i$.
The gap $\score(L[\mb]) - \tilde{F}(L[\mb])$ can be close to 1.
However, the proof of Theorem \ref{thm:near-optimal-approx} only requires $\epsilon_2$ to be at least
the difference between $\score$ and $\tilde{F}$ at position $\mb_l$.
So $\epsilon_2$ is upper-bounded by $\score(L[\mb_l]) - \tilde{F}(L[\mb_l])$ for a given query.

Summarizing the above analysis, the approximation factor $\epsilon$ is determined by the following two factors:
\begin{enumerate}\parskip=0pt
\item how much the approximation $\tilde{F}$ is smaller than the scoring function $\score$ at $\mb_l$,
the starting position of the last segment chosen during the traversal, and
\item how much $\tilde{F}$ is bigger than $\score$ at the optimal position with exactly $|\mb_l|$ steps that
minimizes $\score$.
\end{enumerate}
This yields the following upper bound of $\epsilon$:
\begin{equation}\label{equ:epsilon}
 \epsilon \leq \max\{0, \tilde{\tau} - 1 / \score(L[\mb_l])\} + \score(L[\mb_l]) - \tilde{F}(L[\mb_l]) .
\end{equation}

\section{Experiments on Document and Image datasets}\label{sec:more-experiment}

In addition to the mass spectrometry dataset presented in Section \ref{sec:hull-based}, 
we also verify the near-convexity assumption in a document dataset\footnote{\url{https://www.kaggle.com/ycalisar/hotel-reviews-dataset-enriched}}.
The documents are 515k hotel reviews from the website \url{booking.com}.
We apply the doc2vec\footnote{\url{https://radimrehurek.com/gensim/models/doc2vec.html}} machine learning model to convert each hotel review to a vector representation
and build the inverted index and convex hulls from the generated vectors.
Similar to the mass spectrometry dataset, we observe that the values in the inverted lists have
the similar ``near-convex'' shape as shown in Figure \ref{fig:near-convex-doc}.

\begin{figure}[!ht]
\renewcommand{\tabcolsep}{0.1mm}
\begin{tabular}{ccccc}
\includegraphics[width=0.25\columnwidth]{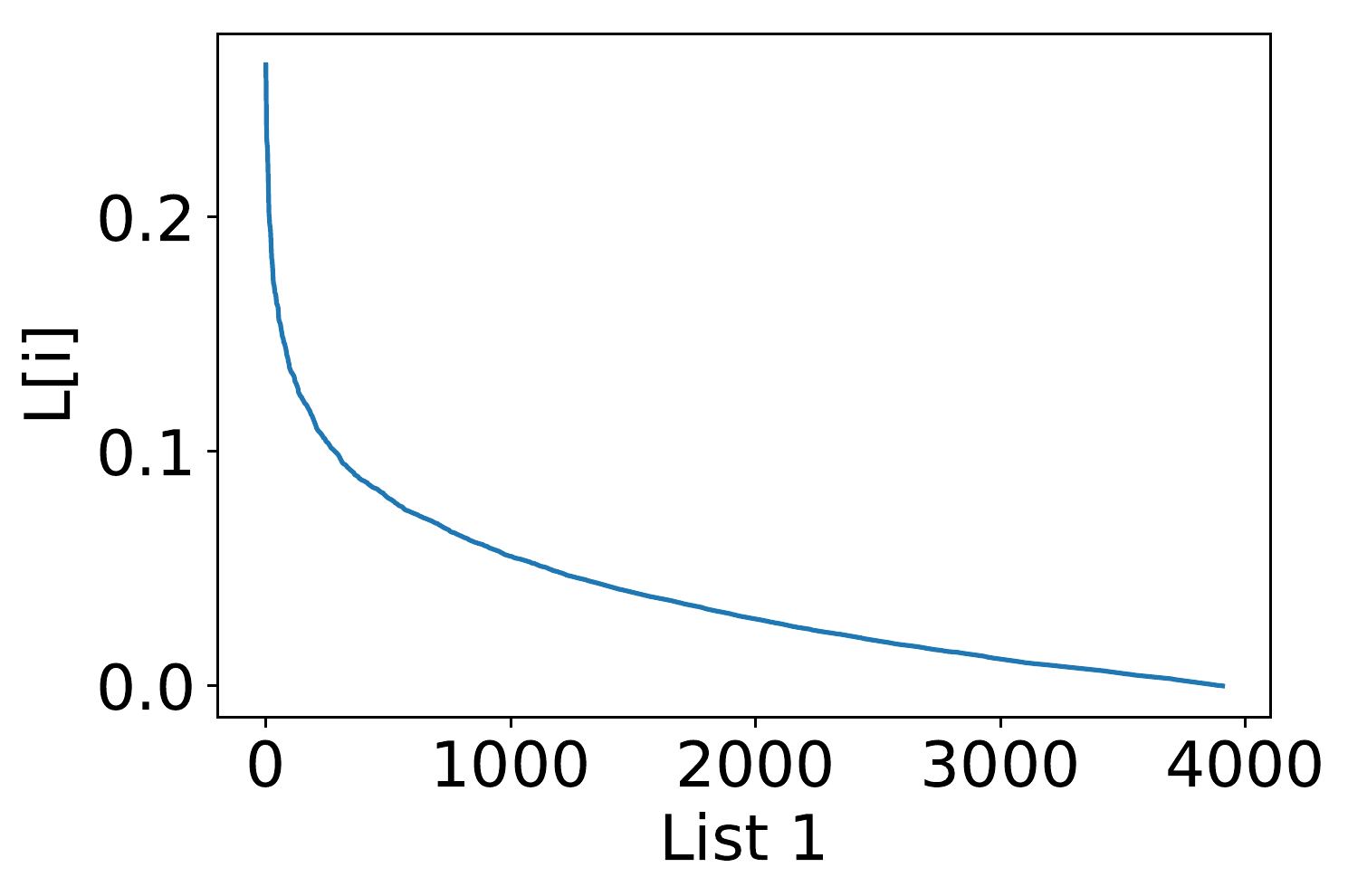}&\includegraphics[width=0.25\columnwidth]{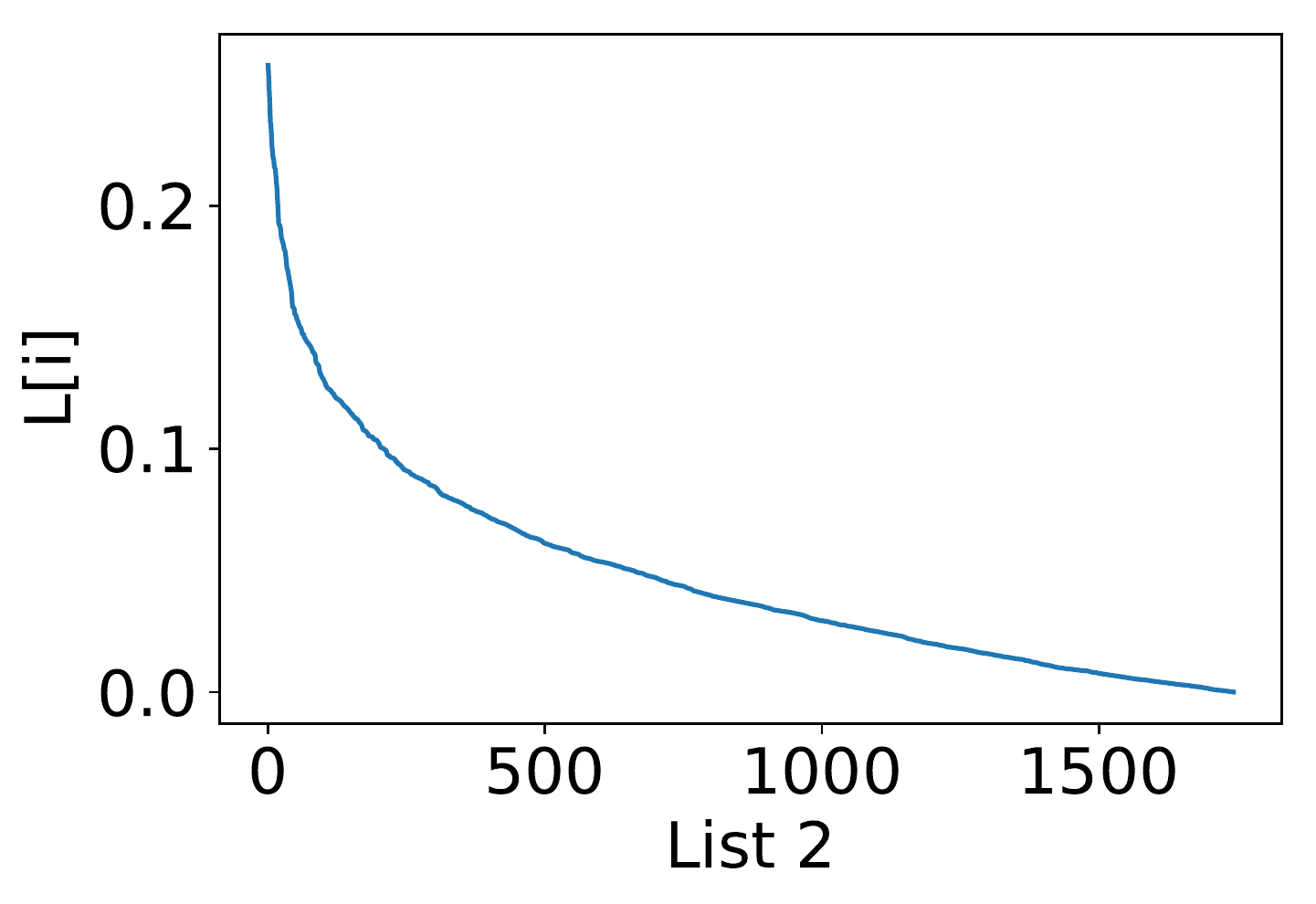}&
\includegraphics[width=0.25\columnwidth]{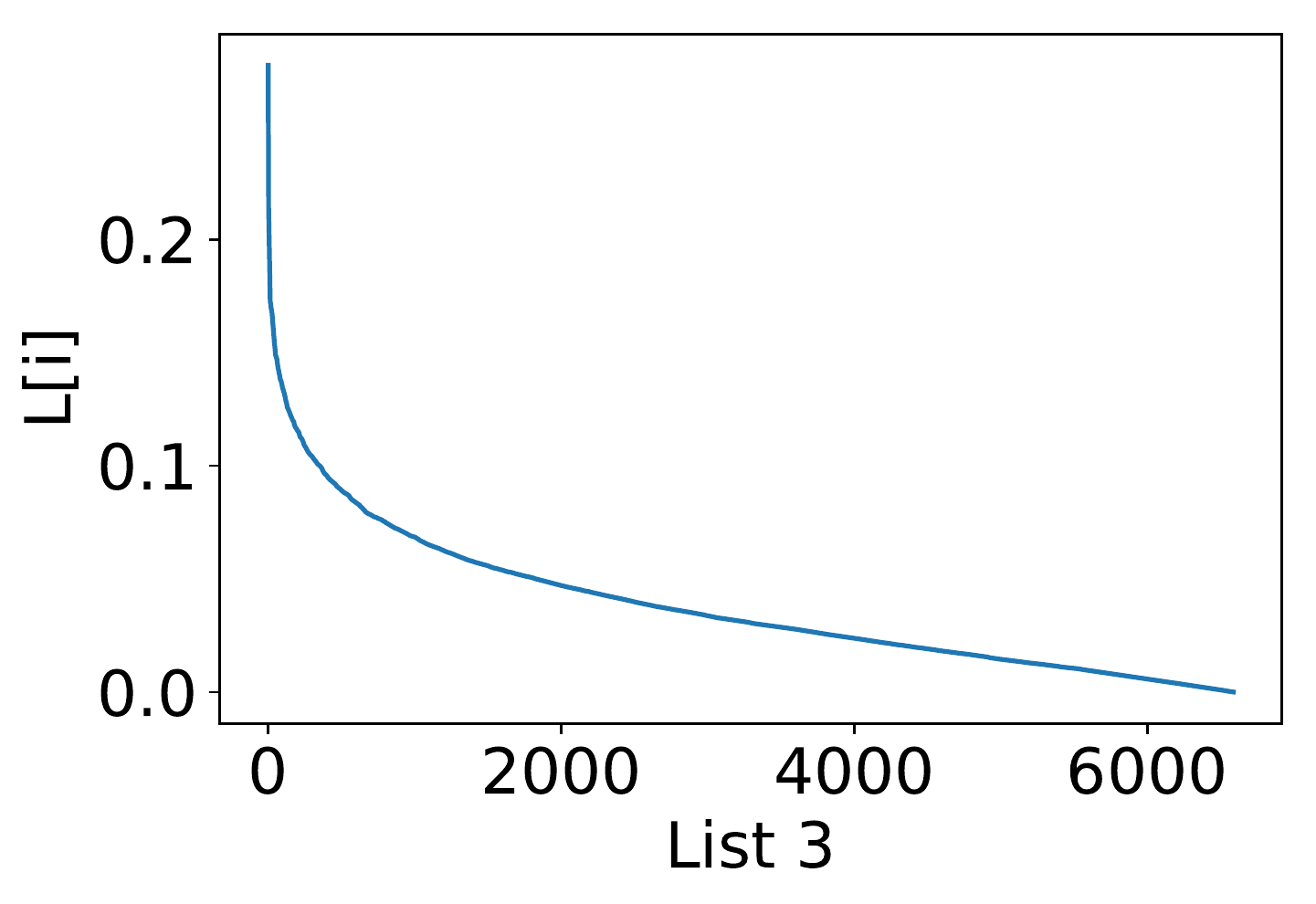}&\includegraphics[width=0.25\columnwidth]{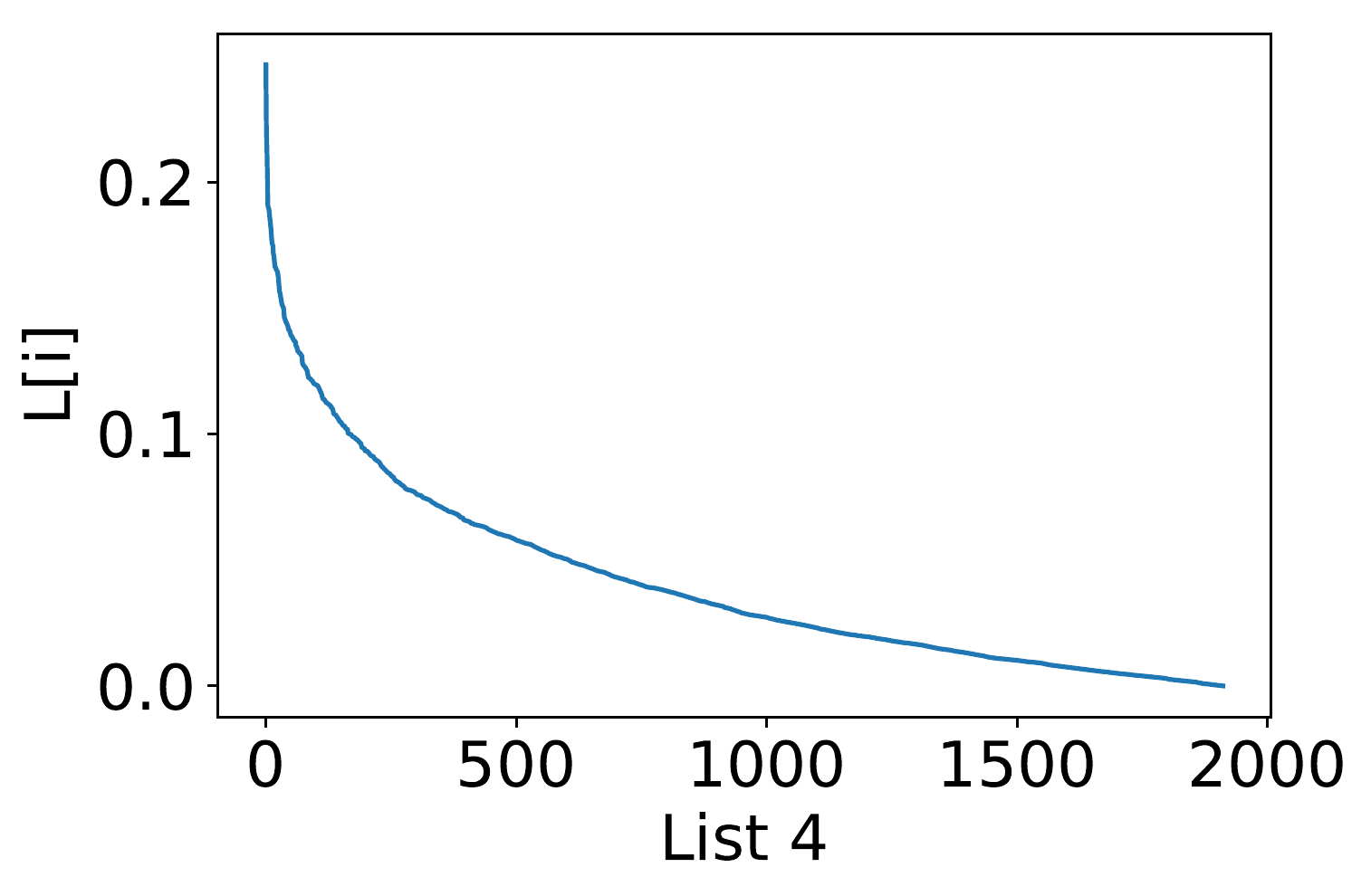}
\end{tabular}
\caption{The inverted lists of the first 4 dimensions of a document dataset}\label{fig:near-convex-doc}
\end{figure}

In document datasets, a cosine threshold query can be interpreted as ``finding similar documents of a given one''.
We ran 100 randomly chosen queries on a subset of 10,000 reviews with threshold $\theta = 0.6$.
The total number of accesses is 4,762,040 and the total sizes of the last gap is 374,521 which
means that the cost additional to the optimal access is no more than 7.86\% 
of the overall access cost.

Next, we repeat the same experiment on an image dataset\footnote{\url{http://vis-www.cs.umass.edu/lfw/}}
containing 13,000 images of human faces collected from the web. We use
the img2vec model\footnote{\url{https://github.com/christiansafka/img2vec}} to 
convert each image to a vector. Once again, we observe the same near-convex shape patterns
in the inverted lists constructed from the vectors as shown in Figure \ref{fig:near-convex-img}.

\begin{figure}[!ht]
\renewcommand{\tabcolsep}{0.1mm}
\begin{tabular}{ccccc}
\includegraphics[width=0.25\columnwidth]{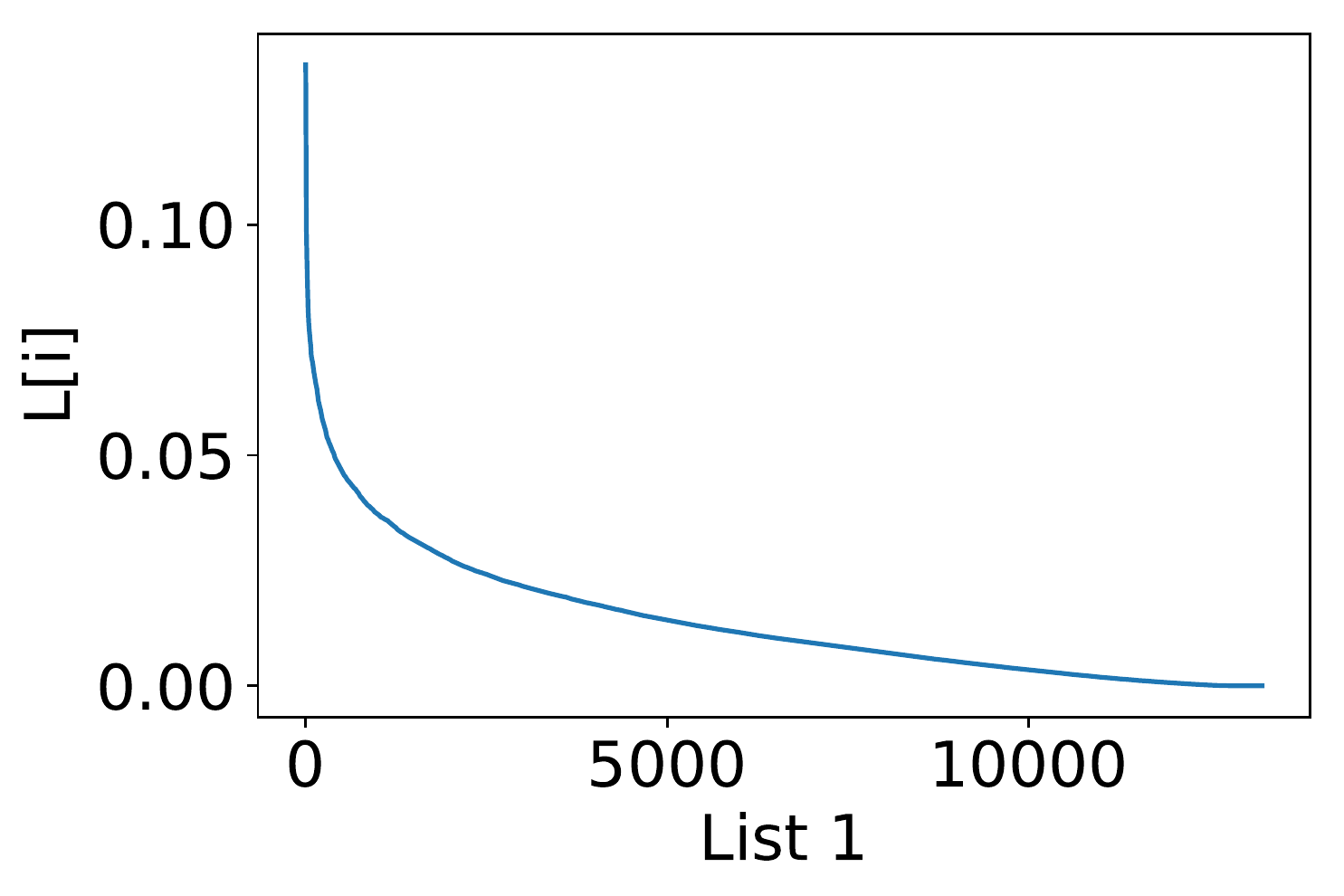}&\includegraphics[width=0.25\columnwidth]{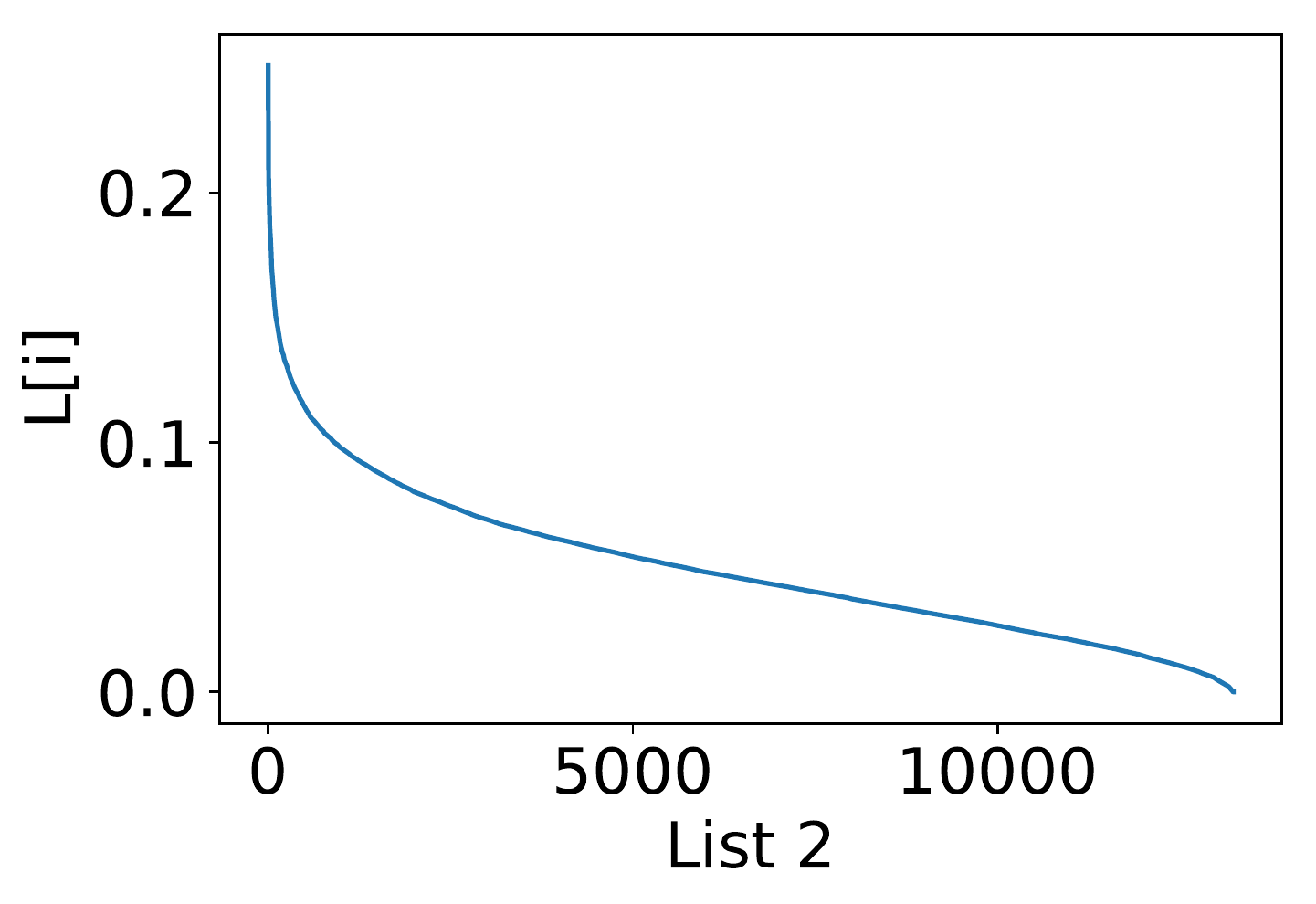}&
\includegraphics[width=0.25\columnwidth]{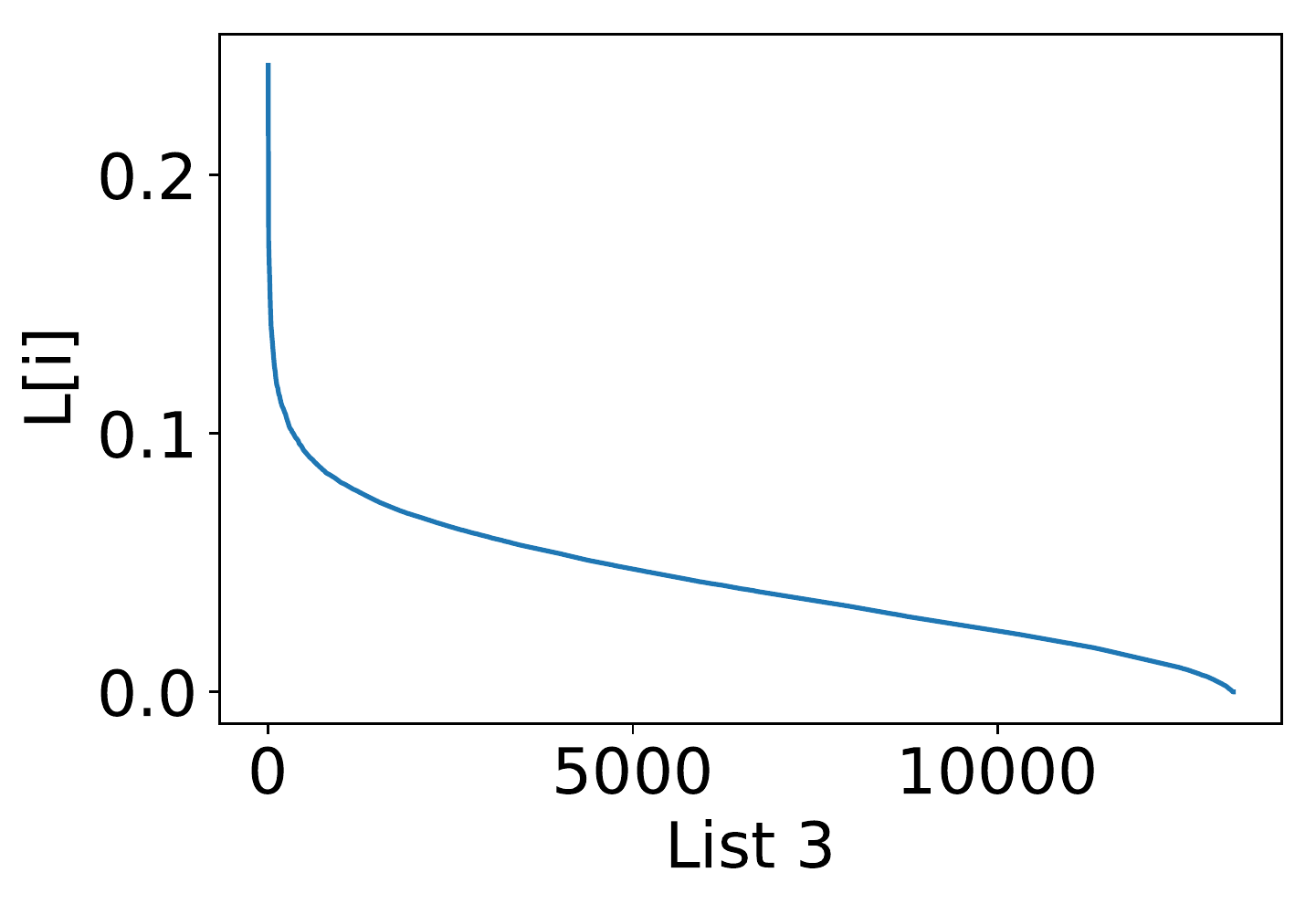}&\includegraphics[width=0.25\columnwidth]{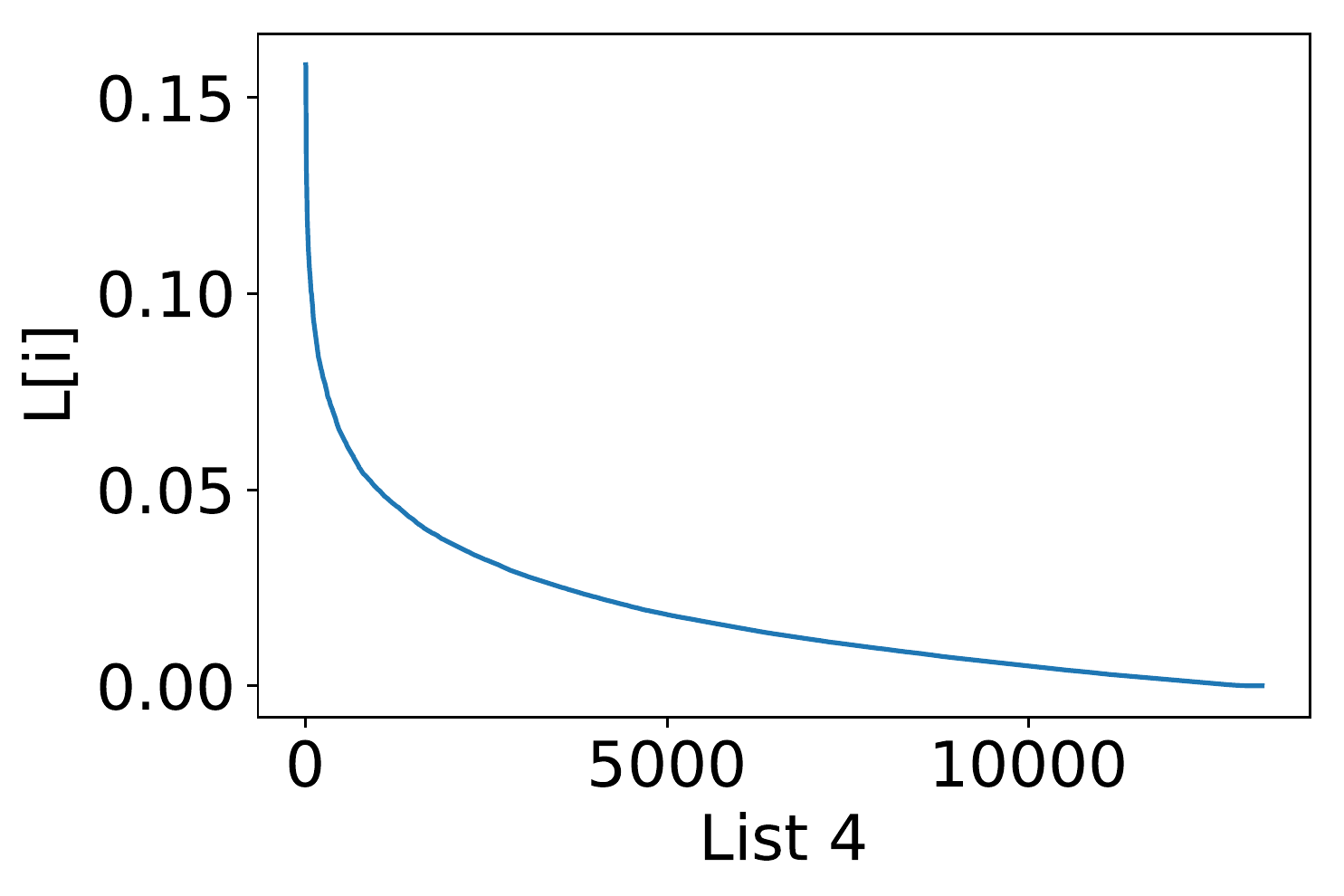}
\end{tabular}
\caption{The inverted lists of the first 4 dimensions of an image dataset}\label{fig:near-convex-img}
\end{figure}

In image datasets, a cosine threshold query can be used for finding similar images to a given image.
With the same setting above, we ran 100 random queries of faces in the whole dataset with the same threshold $\theta = 0.6$.
The total access cost is 118,795,452, the total size of last gaps is 473,999, so
the additional access cost compared to $\opt$ is no more than 0.40\% of the overall access cost.

\section{Applying the proposed techniques to Top-$k$ cosine queries}\label{app:top-k}

We show that the proposed stopping condition and traversal strategy can also be applied to \emph{top-k} cosine queries. 
In the top-k setting, each query is a pair $(\mq, k)$
where $\mq$ is a query vector and $k$ is an integer parameter.
A query $(\mq, k)$ asks for the top $k$ database vectors with the highest cosine similarity with the query vector $\mq$.

The proposed techniques can be applied by adapting the classic structure of $\TA$.
The algorithm traverses the 1-d inverted lists and keeps track of the $k$-th highest similarity score among the gathered vectors. 
Note that to guarantee tightness, the exact similarity scores need to be computed online,
which results in additional computation cost since the exact computation can be avoided for threshold queries
using partial verification (Section \ref{sec:verification}).
The stopping condition can be adapted as follows.
\begin{theorem}
At a position vector $\mb$, let $\theta_k$ be the $k$-th highest score among vectors on or above $\mb$. 
Then the following stopping condition is tight and complete:
$$ \varphi_{\textsf{top-k}}(\mb) = \Big( \score(L[\mb]) < \theta_k \Big) . $$
\end{theorem}
The score $\score(L[\mb])$ can be computed using the same $\calO(\log d)$ incremental maintenance algorithm in Section \ref{sec:incremental}.
The lower bound $\theta_k$ needs to be updated when a new candidate is gathered. 
Computing the similarity score takes $\calO(d)$ time and updating $\theta_k$ can be done in $\calO(\log k)$ using a binary heap.


The hull-based traversal strategy for inner product threshold queries 
can be directly applied to top-k inner product queries. 
The following near-optimality result can be shown.
\begin{theorem}
For a top-k inner product query $(\mq, k)$, 
the access cost of the hull-based traversal strategy $\calT_{\HL}$
on a near-convex database $\db$ is at most $\opt + c$ where 
$c$ is the convexity constant of $\db$.
\end{theorem}

The hull-based strategy can also be applied to top-k cosine queries.
Recall that for cosine threshold queries, the hull-based traversal strategy
operates on the convex hulls of a decomposable approximation $\tilde{F}$ of $\score$ where
$\tilde{F}(L[\mb]) = \sum_{i=1}^d \min\{q_i \cdot \tilde{\tau}, L_i[b_i]\} \cdot q_i $.
As discussed in Section \ref{sec:global-optimal}, the choice of the constant $\tilde{\tau}$
is ideally the value $\tau$ of the optimal traversal path, which is dependent on the threshold $\theta$.
Therefore, for top-k cosine queries where the final threshold $\theta_k$ is unknown in advance,
a bad choice of $\tilde{\tau}$ can lead to poor approximation.
In a practical implementation, the constant $\tilde{\tau}$ can be made query-dependent and more carefully tuned.
We leave these aspects for future work.

\section{Additional Related Work}\label{sec:addRelatedWork}

In this section, we provide additional related work.
\subsection{Approximate approaches}\label{sec:approSearch}

Besides LSH, there are many other approximate approaches for high-dimensional similarity search, e.g., graph-based methods~\cite{FuWC17,Dong2011EKN,Wu2014FUL}, product quantization~\cite{Jegou2011PQN,Andre2015CLE,Lempitsky2012}, randomized KD-trees~\cite{Silpa08}, priority search k-means tree\cite{Muja14}, rank cover tree~\cite{Houle2015}, randomized algorithms~\cite{Yiqiu18},  HD-index~\cite{Arora2018}, and clustering-based methods\cite{Li2002CAS,Cui2010INA}.

\subsection{TA-family algorithms}\label{sec:TAfamily}

Bast et al. studied top-k ranked query processing with a different goal of minimizing the overall cost by scheduling the best sorted accesses and random accesses~\cite{Bast2006IIO}. However, when the number of dimensions is high, it requires a large amount of online computations to frequently solve a Knapsack problem,
which can be slow in practice. Note that, \cite{Bast2006IIO} only evaluated the number of sorted and random accesses in the experiments instead of the wall-clock time.  Besides, it does not provide any theoretical guarantee, which is also applicable to ~\cite{Jin2011EGE}.
Akbarinia et al. proposed \textsf{BPA} to improve \textsf{TA}, but the optimality ratio is the same as \textsf{TA} in the worst case~\cite{Akbarinia2007BPA}.
Deshpande et al. solved a special case of top-k problem by assuming that the attributes are drawn from a small value space~\cite{Deshpande2008}.
Zhang et al. developed an algorithm targeting for a large number of lists~\cite{Zhang2016LAT} by merging lists into groups and then apply $\TA$. However, it requires the ranking function to be distributive that does not hold for the cosine similarity function.
Yu et al.  solved top-k query processing in subspaces~\cite{Yu2014} that are not applicable to general cosine threshold queries.

Some works considered non-monotonic ranking functions~\cite{Zhang2006BRQ,Xin2007PSM}. For example, \cite{Zhang2006BRQ} focused on the combination of a boolean condition with a regular ranking function and \cite{Xin2007PSM} assumed the ranking functions are lower-bounded.
The cosine function has not been considered in this line of work. 

\subsection{CSS in Hamming Space}\label{sec:hamming}
\cite{EghbaliT16} and \cite{Jianbin18} studied cosine similarity search in the Hamming space where each vector contains only binary values. 
They proposed to store the binary vectors efficiently into multiple hash tables.
However, the techniques proposed there cannot be applied since transforming real-valued vectors to binary vectors loses information, and thus correctness is not guaranteed.

\yuliang{moved inner product search here.}

\subsection{Inner product search}\label{sec:relatedproblems}



The cosine threshold querying is also related to inner product search where vectors may not be unit vectors, otherwise, inner product search is equivalent to cosine similarity search.


%

Teflioudi et al. proposed the \textsf{LEMP} framework~\cite{Teflioudi2015LFR,Teflioudi2016} to solve the inner product search where each vector may not be normalized.
The main idea is to partition the vectors into buckets according to vector lengths and then apply
an existing cosine similarity search (CSS) algorithm to each bucket.
This work provides an efficient way for CSS that can be integrated to the \textsf{LEMP} framework.

Li et al. developed an algorithm \textsf{FEXIPRO}~\cite{Li2017FFE} for inner product search in recommender systems
where vectors might contain negative values. Since we focus on non-negative values in this work, the proposed techniques
(such as length-based filtering and monotonicity reduction) are not directly applicable.

There are also tree-based indexing and techniques for inner product search~\cite{Ram2012MIS,CurtinGR13}. But they are not scalable to high dimensions~\cite{Keivani17}. 
Another line of research is to leverage machine learning to solve the problem~\cite{MussmannE16,Fraccaro2016IPM,Shen15}. However, they do not provide accurate answers. Besides, they do not have any theoretical guarantee.

\subsection{Dimensionality reduction}\label{app:dimensionality-reduction}
Since many techniques are affected negatively by the large number of dimensions, one potential solution is to apply dimensionality reduction
(e.g. PCA, Johnson-Lindenstrauss)~\cite{Liaw2010FEK,LianC09,Johnson1986} before any search algorithm. 
However, this does not help much if the dimensions are not correlated and there are no hidden latent variables to be uncovered by dimensionality reduction. For example, in the mass spectrometry domain, each dimension represents a chemical compound/element and there is no physics justifying a correlation. We applied dimensionality reduction to the data vectors and turned out that only 4.3\% of dimensions can be removed in order to preserve 99\% of the distance.

\subsection{Keyword search}\label{sec:kwdsearch}




This work is different from keyword search although the similarity function is (a variant of) the cosine function and the main technique is inverted index~\cite{M08}. There are two main differences: (1) keyword search generally involves a few query terms~\cite{Wang2017MIL}; (2) keyword search tends to return Web pages that contain all the query terms~\cite{M08}. As a result, search engines (e.g., Apache Lucene) mainly sort the inverted lists by document IDs~\cite{Broder2003EQE,Ding2011FTD,Chakrabarti11} to facilitate boolean intersection. Thus, those algorithms cannot be directly applied to cosine threshold queries.
Although there are cases where the inverted lists are sorted by document frequency (or score in general) that are similar to this work, they usually follow  \textsf{TA}~\cite{Ding2011FTD}. This work significantly improves \textsf{TA} for cosine threshold queries.


\subsection{Mass spectrometry search}\label{sec:masssearch}

%
%


For mass spectrometry search, the state-of-the-art approach is to partition the spectrum database
into buckets based on the spectrum mass (i.e., molecular weight) and only search the buckets that have the similar mass to the query~\cite{MSFragger}.
Each bucket is scanned linearly to obtain the results.
The proposed techniques in this work can be applied to each bucket and improve the performance dramatically.
\yuliang{I moved the related work on mass spec from the main text to here.}

Library search, wherein a spectrum is compared to a series of reference spectra to find the most similar match, providing a putative identification for the query spectrum~\cite{ac980122y}. Methods in library search often use cosine~\cite{PMIC200600625} or cosine-derived~\cite{pr400230p} functions to compute spectrum similarity.  Many of the state-of-the-art library search algorithms, e.g., SpectraST~\cite{PMIC200600625}, Pepitome~\cite{Dasari2012}, X!Hunter~\cite{Craig2006}, and \cite{Mingxun2016} find candidates by doing a linear scan of peptides in the library that are within a given parent mass tolerance and computes a distance function against relevant spectra.  This is likely because the libraries that are searched against have been small, but spectral libraries are getting larger, e.g.,  MassIVE-KB now has more than 2.1 million precursors.  This work fundamentally improves on these, as it uses inverted lists with many optimizations to significantly reduce the candidate spectra from comparison.  M-SPLIT used both a prefiltering and branch and bound technique to reduce the candidate set for each search~\cite{Wang2010}. This work improves on both these by computing the similar spectra directly, while still applying the filtering.


Database search~\cite{Eng1994} is where the experimental spectrum is compared to a host of theoretical spectra that are generated from a series of sequences.  The theoretical spectra differ from experimental spectra in that they do not have intensities for each peak and contain all peaks which could be in the given peptide fragmentation. In other words, the values in the vectors are not important and the similarity function evaluates the number of shared dimensions that are irrelevant to the values.  While the problem is different in practice than that described in this paper, it is still interesting to consider optimizations. Dutta and Chen proposed an LSH-based technique which embeds both the experimental spectra and the theoretical spectra onto a higher dimensional space~\cite{Dutta07}. While this method performs well at the expense of low recall rate, our method is guaranteed to not lose any potential matches. There are also other methods optimizing this problem by using preindexing~\cite{ac0481046,MSFragger} but none do this kind of preindexing for calculating matches between experimental spectra.







\subsection{Others}
This work is also different from ~\cite{Wenhai18} because that work focused on similarity search based on set-oriented p-norm similarity, which is different from cosine similarity. 



In the literature, there are skyline-based approaches~\cite{Lee2012EDL} to solve a special case of top-k ranked query processing where no any ranking function is specified. However, those approaches cannot be used to solve the problems that only involves \emph{unit vectors} such that all the vectors belong to skyline points.

\end{document}